\documentclass[twocolumn]{autart}
\usepackage{amsfonts,amsmath,amssymb,bm,mathtools}  
\usepackage{graphicx,subcaption,epstopdf,pdfpages}  
\usepackage{enumitem}
\usepackage{comment,soul,xcolor}                    
\allowdisplaybreaks

\newtheorem{theorem}{Theorem}
\newtheorem{algorit}[thm]{Algorithm} 
\newtheorem{lemma}[thm]{Lemma}
\newtheorem{proposition}[thm]{Proposition}

\theoremstyle{definition}
\newtheorem{definition}[thm]{Definition}
\newtheorem{example}[thm]{Example}
\newtheorem{remark}[thm]{Remark}

\newcommand{\q}{q^{-1}}
\newcommand{\z}{z^{-1}}
\newcommand{\en}{\hat\eta_N}
\newcommand{\tn}{\hat\theta_N}
\newcommand{\vn}{\hat\vartheta_N}
\newcommand{\zn}{\hat\zeta_N}
\newcommand{\bn}{\hat\beta_N}
\newcommand{\E}{\mathbb{E}}
\newcommand{\Eb}{\bar{\E}}
\newcommand{\Four}{\mathfrak{F}}

\newcommand{\rank}{\text{rank}}
\newcommand{\eio}{e^{i\omega}}

\begin{document}

\begin{frontmatter}
    \title{Modelling and identification of diffusively coupled linear networks with additional directed links\thanksref{funding}} 
    \thanks[funding]{Funded by the European Union. Views and opinions expressed are however those of the author(s) only and do not necessarily reflect those of the European Union or the European Research Council. Neither the European Union nor the granting authority can be held responsible for them. This paper was not presented at any IFAC meeting. }

    \author[TUe,DIFFER]{E.M.M. (Lizan) Kivits\corauthref{cor}}\ead{e.m.m.kivits@differ.nl},
    \author[TUe]{Paul M.J. Van den Hof}\ead{p.m.j.vandenhof@tue.nl},
    \corauth[cor]{Corresponding author: E.M.M. (Lizan) Kivits. }

    \address[TUe]{Department of Electrical Engineering, Eindhoven University of Technology, 5600 MB Eindhoven, the Netherlands}
    \address[DIFFER]{DIFFER - Dutch Institute for Fundamental Energy Research, De Zaale 20, 5612 AJ Eindhoven, the Netherlands\thanksref{nwo}}
    \thanks[nwo]{DIFFER is part of the institutes organization of NWO. }
    
    \begin{keyword} 
        Modelling; System identification; Identification methods; Dynamic networks; Physical systems; Diffusive couplings; Identifiability
    \end{keyword} 

    \begin{abstract} 
    Dynamic networks consist of interconnected dynamical systems. 
    The subsystems can be viewed as transformations of input signals into output signals, where signals flow from one system into another through interconnections. The signal flows represent directions of information flow, thus a dynamic network can be visualised by a directed graph. 
    In contrast, natural and physical laws only impose relations between systems variables, while variables are shared among systems via interconnections. Sharing is independent of direction, and therefore a dynamic network originating from physics can be visualised by an undirected graph. 
    Typically, dynamic networks are considered to have either directed or undirected interconnections. For both situations, network models, analytic tools, and identification algorithms have been developed.
    However, dynamic networks can also have both directed and undirected interconnections, for example, in physical networks equipped with digital controllers.
    In this work, we present mixed linear dynamic networks that contain both undirected and directed interconnections, where the nature of the interconnecting dynamics needs to be incorporated into the modelling framework, identifiability analysis, and identification procedure. For these mixed networks, we derive dynamic network models; formulate conditions for consistent identification of all dynamics in the network; and develop a tractable identification algorithm that delivers consistent estimates.
    \end{abstract}

\end{frontmatter}

\section{Introduction} \label{section:introduction}
Dynamic networks have been studied in a variety of research domains including social science, finance, biology, and engineering \cite{BOCCALETTI2006,MESBAHI2010,REN2005}. In the past decade, modelling and identification of dynamic networks has received significant attention in the literature. From an identification point of view, a dynamic network is a natural successor of the classical open-loop and closed-loop systems, while more complex interconnection structures are considered. The large-scale systems that are modelled with block diagrams in Matlab Simulink \cite{Simulink} are examples of dynamic networks. This viewpoint led to a modelling framework for dynamic networks in which transfer functions (modules) are interconnected to describe the dependencies among signals \cite{GONCALVES2008,VANDENHOF2013}. Dynamic networked systems can also be described by interconnected state-space systems \cite{VERHAEGEN2012,ZHOU2022}. Because of the input-output structure of the modules, such dynamic networks can graphically be depicted by a directed graph, in which the signals are represented by the nodes and where the edges capture the modules. For these directed dynamic networks, a prediction error identification framework has been developed \cite{VANDENHOF2013}. A number of identification topics has been addressed, including identification of the full network \cite{GONCALVES2008,MATERASSI2020,WEERTS2018_AUTO2} and identifiability of the network dynamics under various assumptions on the measured and excited signals \cite{CHENG2022,HENDRICKX2019,LEGAT2020,SHI2023,VANWAARDE2020,WEERTS2018_AUTO1}, even for networks including fixed dynamics \cite{DREEF2022} or nonlinear functions \cite{VIZUETE2023}. 

Physical processes have a wide range of applications. Think, for example, of electrical circuits, chemical reactions, and mechanical, hydraulic, and biological systems \cite{LJUNG1994}. Physical laws can often be expressed by ordinary differential equations, which are therefore used in many models for physical systems, such as state-space models \cite{BULLO2022,HABER2014,MESBAHI2010,YU2020}, second-order vector differential equation models \cite{LJUNG1994}, and polynomial models \cite{LJUNG1999,LJUNG1994}. A physical network consists of interconnected physical components, which often imply couplings based on the difference of variables and are characterised by symmetric behaviour. Interconnections that depend on the difference of the connected variables are referred to as diffusive couplings \cite{JONES1985}. For example, the behaviour of a linear resistor is symmetric in the sense that the current flowing through the resistor depends \emph{equally} on the electric potentials on each side of the resistor. Moreover, it depends on the \emph{difference} between these electric potentials, implying a diffusive coupling. The resistor thus implies a symmetric cause-effect relationship instead of an input-output structure. Consequently, diffusively coupled networks can graphically be represented by undirected graphs, akin the schematics in Matlab Simscape \cite{Simscape} that are used to model physical processes. 

For identification purposes, diffusively coupled networks can be modelled in several frameworks. The modelling frameworks of directed networks is less attractive for including symmetric diffusive couplings, because their structural properties are more difficult to incorporate into the modelling and identification procedure and cannot be accounted for in the analysis. State-space forms \cite{LOPES2015,LUS2003} and transfer function models or polynomial models \cite{LJUNG1999} can be used to estimate the model parameters of diffusively coupled networks, but in general the interconnection structure in the network is lost in the model. To preserve the internal interconnection structure and to include the characteristic symmetric property of the dynamics in the model, diffusively coupled networks can be modelled in a particular polynomial framework \cite{KIVITS2023_TAC}. This model is the basis for prediction error identification methods that are able to incorporate the characteristics into the identification procedure for identifying the full dynamics and topology of the network \cite{KIVITS2023_TAC,LIANG2025} or for identifying a particular subnetwork \cite{KIVITS2022}. However, the consequences of including directed interconnections in this framework are not clear yet. 

Due to the symmetric cause-effect relationships in the interactions, physical networks can be represented by undirected graphs. On the other hand, digital controllers explicitly describe input-output relationships and therefore, can only be represented by directed graphs. The same holds for nonsymmetric physical components, such as diodes and check valves. Physical processes that contain nonsymmetric physical components or that are equipped with digital controllers contain a mix of directed and undirected (diffusive) connections. Dynamic networks containing both undirected and directed interconnections are referred to as \emph{mixed networks} \cite{KIVITS2024t}. The objective of this research is to identify the dynamics of mixed linear dynamic networks. This involves determining an attractive modelling framework and selecting the signals that need to be measured and\slash or excited to identify all or some of the network dynamics. In addition, the consistency and minimum variance properties of the estimates have to be specified and algorithms for performing the identification have to be developed and implemented. These objectives can be achieved by extending the theory for modelling and identification of diffusively coupled networks to include directed dynamics.

This paper includes the modelling frameworks and the theoretical identification results that lead towards achieving the above-mentioned objectives. Section~\ref{section:network} contains an introduction to mixed networks. In Section~\ref{section:modelling}, the network model of the particular class of mixed networks that will be analysed is formalised. The identification problem is defined in Section~\ref{section:identification}. Conditions for a unique representation (i.e. network identifiability) are presented in Section~\ref{section:identifiability}, while conditions for data informativity and consistent identification are presented in Section~\ref{section:consistency}. In Section~\ref{section:algorithm}, the identification algorithm is discussed. Further insights and extensions are discussed in Section~\ref{section:discussion}, after which Section~\ref{section:conclusion} concludes the paper.

Consider the following notation throughout the paper: 
Let $\mathbb{R}^{p,m}[\z]$ be the ring of matrix polynomials of size $p \times m$ in complex indeterminate $\z$. A polynomial matrix $A(\z)\in\mathbb{R}^{p,m}[\z]$ consists of matrices $A_{\ell}$ and $(j,k)$th polynomial elements $a_{jk}(\z)$ such that $A(\z) = \sum_{\ell=0}^{n_a} A_{\ell} z^{-\ell}$ and $a_{jk}(\z) = \sum_{\ell=0}^{n_a} a_{jk,\ell} z^{-\ell}$. Hence, $\lim_{z\rightarrow\infty} A(z) = A_0$. 
Let $\mathbb{R}^{p\times m}(z)$ be the field of rational matrix functions of size $p \times m$. A rational function matrix $F(z)\in\mathbb{R}^{p,m}(z)$ consists of elements $F_{jk}(z)$ and is proper if $F^{\infty} := \lim_{z\rightarrow\infty} F(z) = c \in \mathbb{R}^{p\times m}$; it is strictly proper if $F^{\infty} = 0$, and monic if $p=m$ and $F^{\infty}=I$, the identity matrix. $F(z)$ is stable if all its poles are within the unit circle $|z| < 1$. A matrix is called \emph{hollow} if all elements on the diagonal are zero. 
$gcd(x,Y)$ is the greatest common divisor of scalar $x$ and all scalar elements of matrix $Y$.

\section{Controlled physical networks} \label{section:network}
In this paper, a linear dynamic network is considered to describe the dynamical relations between internal signals, referred to as node signals. The node signals can be affected by measured external excitation signals and unmeasured disturbance signals. 

Physical networks are interconnections of physical components, which can express physical processes from various domains. Common examples include electrical resistor-inductor-capacitor circuits and mechanical mass-spring-damper systems. Physical networks typically contain a connection to a reference node or ground, which is, for example, an electric component attached to the earth or a mechanical component secured to the wall. The ground can be seen as a node with node signal $w_0(t)=0$. The symmetric bi-directional nature of physical components leads to symmetric cause-effect relationships. This characteristic property is captured by diffusive couplings \cite{JONES1985} and plays a key role in the model. 

\begin{figure}
    \centering
    \includegraphics[page=2,width=\linewidth]{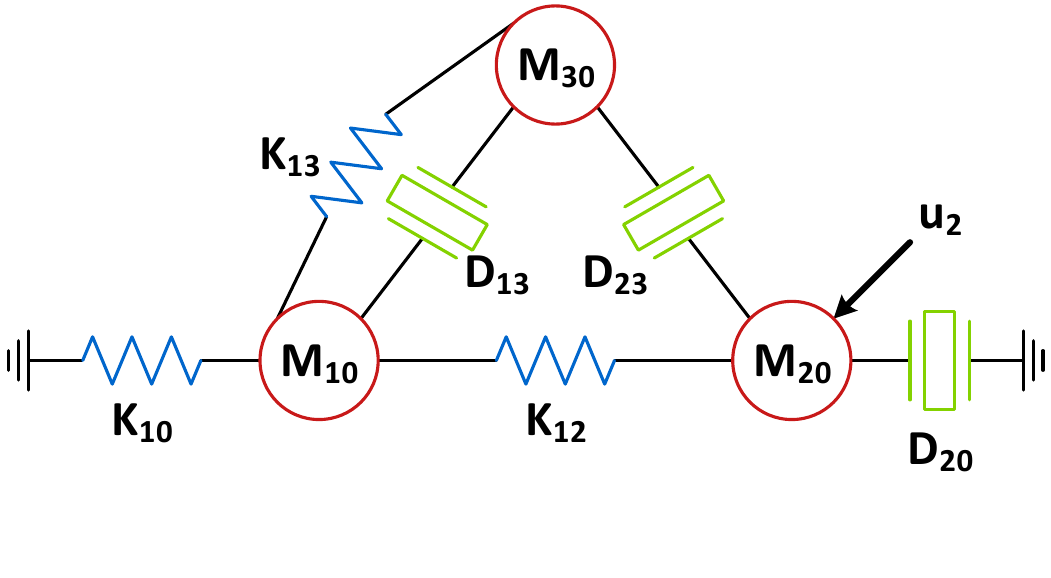}
    \caption{A network of masses ($M_{j0}$) interconnected through dampers ($D_{jk}$) and springs ($K_{jk}$), with input ($u_2$) and a ground (triangle of three lines) \cite{KIVITS2023_TAC}.}
    \label{fig:msd}
\end{figure}

An example of a physical network is the mass-spring-damper system shown in Fig.~\ref{fig:msd} \cite{KIVITS2023_TAC}. It consists of three masses $M_{j0}$, $j=1,2,3$, that are being interconnected with dampers $D_{jk}$ and springs $K_{jk}$, with $k\neq0$, and that are being attached to the ground (or wall) with dampers $D_{j0}$ and springs $K_{j0}$. The couplings between the masses are diffusive, because springs and dampers are symmetric components. The positions $w_j(t)$ of the masses $M_{j0}$ are the internal signals that are considered as the node signals\footnote{A system as the one shown in Fig.~\ref{fig:msd} would require at least a two-dimensional position vector $w_j(t)$, but for notational convenience and without loss of generality, we will restrict our attention to scalar-valued node
signals $w_j(t)$.}. The behaviour of each node signal is described according to the second-order differential equation
\begin{multline}\label{eq:msd_node}
    M_{j0} \ddot{w}_j(t) + D_{j0} \dot{w}_j(t) + \sum_{k\in\mathcal{N}_j} D_{jk} \left[ \dot{w}_j(t) - \dot{w}_k(t) \right] \\
    + K_{j0} w_j(t) + \sum_{k\in\mathcal{N}_j} K_{jk} \left[ w_j(t) - w_k(t) \right] = u_j(t),
\end{multline}
with $\mathcal{N}_j$ the set of indices of node signals $w_k(t)$, $k\neq j$, with connections to node signals $w_j(t)$; external signals $u_j(t)$; and $\dot{w}_j(t)$ and $\ddot{w}_j(t)$ the first- and second-order derivative of $w_j(t)$, respectively. The overall behaviour of this system is captured by
\begin{equation} \label{eq:msd_diffusive}
    M_d \ddot{w}(t) + (D_d + D_{\mathcal{L}}) \dot{w}(t) + (K_d + K_{\mathcal{L}}) w(t) = u(t),
\end{equation}
with $w(t)$ and $u(t)$ being vectorised versions of $w_j(t)$ and $u_j(t)$, respectively; and with $M_d$, $D_x$, and $K_x$, $x\in\{d,\mathcal{L}\}$ containing the masses, the damping coefficients, and the spring coefficients, respectively, as follows: Diagonal matrices $M_d$, $D_d$, and $K_d$ are having $(j,j)$th elements $M_{j0}$, $D_{j0}$, and $K_{j0}$, respectively; and Laplacian\footnote{A Laplacian matrix is a symmetric matrix with nonpositive off-diagonal elements and with nonnegative diagonal elements that are equal to the negative sum of all other elements in the same row (or column) \cite{MESBAHI2010}.} matrices $D_{\mathcal{L}}$ and $K_{\mathcal{L}}$ are having off-diagonal $(j,k)$th elements $[D_{\mathcal{L}}]_{jk} =D_{jk}$ and $[K_{\mathcal{L}}]_{jk} =K_{jk}$, $j\neq k$, respectively, and diagonal $(j,j)$th elements $[D_{\mathcal{L}}]_{jj} = -\sum_{k\neq j}D_{jk}$ and $[K_{\mathcal{L}}]_{jj} = -\sum_{k\neq j}K_{jk}$, respectively. The diffusive type of coupling induces the symmetry constraints $D_{jk} = D_{kj}$ and $K_{jk} = K_{kj}$ $\forall j,k$.

Let a digital controller be added to the mass-spring-damper system between masses $M_{20}$ and $M_{10}$, e.g. taking the position of mass $M_{20}$ as input to steer the position of mass $M_{10}$. The resulting controlled physical network is shown in Fig.~\ref{fig:msdc} and is described by
\begin{multline} \label{eq:msd_mixed}
    M_d \ddot{w}(t) + (D_d + D_{\mathcal{L}}) \dot{w}(t) + (K_d + K_{\mathcal{L}}) w(t) \\= u(t) + G w(t),
\end{multline}
with $G$ a matrix containing the proportional gain controller $G_{12}$ as its $(1,2)$th element and with zeros at the other elements. If $G_{12}$ is known, it acts like input dynamics on the network, even though its input signal is coming from inside the network. However, if the dynamics of this controller is unknown, the model has to be adjusted to account for the directed dynamics.  

\begin{figure}
    \centering
    \includegraphics[page=3,width=\linewidth]{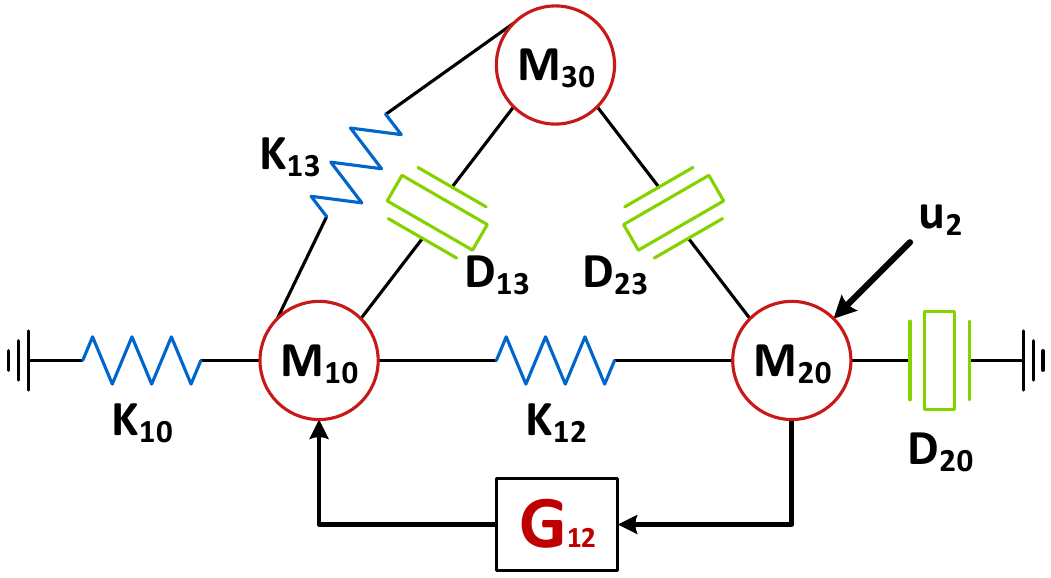}
    \caption{The mass-spring-damper network of Fig.~\ref{fig:msd} with an additional digital controller ($G_{12}$).}
    \label{fig:msdc}
\end{figure}

\section{Network models} \label{section:modelling}

\subsection{Diffusively coupled network model} \label{subsection:diffusive}
As mentioned, the symmetric bi-directional nature of physical components is captured by diffusive couplings. For modelling diffusively coupled networks, we adopt the polynomial framework defined in \cite{KIVITS2023_TAC}. To fully exploit the corresponding results on network identification, we consider the network model in discrete-time. In \cite{KIVITS2023_TAC}, the discrete-time model is derived from the continuous-time model by using the backward difference discretisation method. As a result, the continuous-time and discrete-time model of diffusively coupled networks have the same structural properties in capturing the diffusive couplings. 
\begin{definition}[Diffusively coupled network model] \label{def:diffusive}
  A diffusively coupled network consists of $L$ node signals $w_j(t)$, $j=1,2,\ldots,L$, and $K$ excitation signals $r_k(t)$, $k=1,2,\ldots,K$, and is defined as
  \begin{equation}\label{eq:diffusive}
    A(\q) w(t) =  B(\q) r(t) + F(q) e(t),
  \end{equation}
  with $q^{-1}$ the shift operator meaning $q^{-1}w(t) = w(t-1)$ and with
  \begin{enumerate}
    \item $A(\q)\in \mathcal{A}:=\{A\in\mathbb{R}^{L\times L}[\q]\ |\ a_{jk} = a_{kj}\linebreak \forall k,j\ \mbox{and}\ A^{-1}\ \mbox{stable} \}$.
    \item $B(\q)\in \mathcal{B}:=\{B\in\mathbb{R}^{L\times K}[\q]\}$.
    \item $F(q) \in \mathcal{F}:=\{F\in\mathbb{R}^{L\times L}(q)\ |\ F\ \mbox{monic, stable, and}\linebreak \mbox{stably invertible} \}$.
    \item $\Lambda\succ0$ the covariance matrix of the noise $e(t)$.
    \item $w(t)\in\mathbb{R}^{L\times 1}$ a column vector collecting all node signals $w_j(t)$, $j=1,2,\ldots,L$.
    \item $r(t)\in\mathbb{R}^{K\times 1}$ a deterministic and bounded sequence.
    \item $e(t)\in\mathbb{R}^{L\times 1}$ a zero-mean white noise process with bounded moments of an order larger than 4 \cite{LJUNG1999}.
  \end{enumerate}
  This network is assumed to be connected with at least one connection to the ground\footnote{If the network is connected (each pair of nodes has a path connecting them), it has at least one connection to the ground if and only if $\rank( A(\z) ) = L$ \cite{DORFLER2013,KIVITS2023_TAC}.}.
\end{definition}

The internal dynamics of the network are described by polynomial matrix $A(\q)$. The symmetry of the diffusive couplings appear in \eqref{eq:diffusive} via the symmetry of $A(\q)$. Note that $A(\q)$ is also nonmonic, which is in contrast to the standard ARMAX model structure \cite{HANNAN2012}. To interpret the network model \eqref{eq:diffusive} in terms of physical components, observe that $A(\q)$ can uniquely be separated as $A(\q)=X(\q)+Y(\q)$, where $X(\q)$ is diagonal and $Y(\q)$ is Laplacian. The diagonal elements of $X(\q)$ capture the dynamics in the connections of the nodes $w(t)$ to the ground. The off-diagonal elements of $Y(\q)$ capture the dynamics in the interconnections between the nodes $w(t)$. The diffusive character of the components appears in $Y(\q)$ as the diagonal element in each row being equal to the negative sum of all other elements in that row, that is $y_{jj}(\q)=-\sum_{k\neq j}y_{jk}(\q)$. The symmetry of $Y(\q)$ represents the symmetry properties of the physical components. The assumption on connectivity induces well-posedness of the network. 

\begin{example}[Diffusively coupled network model] \label{ex:msd_diffusive}
    The (continuous-time) model \eqref{eq:msd_diffusive} of the mass-spring-damper system described in Section~\ref{section:network} has the same structure as the (discrete-time) diffusively coupled network model in \eqref{eq:diffusive} with $u(t):=B_0 r(t) + v(t)$ and with symmetric matrices $A_2=M_d$, $A_1=D_d+D_{\mathcal{L}}$, $A_0=K_d+K_{\mathcal{L}}$. The symmetric matrices $A_{\ell}$, $\ell=0,1,2$, can be separated into diagonal matrices $X_2=M_d$, $X_1=D_d$, and $X_0=K_d$, and Laplacian matrices $Y_1=D_{\mathcal{L}}$, and $Y_0=K_{\mathcal{L}}$. 
\end{example}

\subsection{Mixed network model} \label{subsection:mixed}
Although most physical components are symmetric, some components exhibit nonsymmetric or even directed behaviour. Components like diodes and check valves introduce directionality due to their one-way conduction. The same holds for input-output systems, such as digital controllers. Physical processes that contain nonsymmetric physical components or that are equipped with digital controllers contain a mix of directed and undirected (diffusive) connections. The diffusively coupled network model as defined in Definition~\ref{def:diffusive} cannot include components with diffusive nonsymmetric behaviour, components with nondiffusive behaviour, and input-output systems, because these directed systems ruin the symmetric structure of the diffusively coupled network model. Dynamic networks in which undirected and directed interconnections coexists are referred to as mixed networks. 
\begin{definition}[Mixed network] \label{def:mixednetwork}
  A \emph{mixed network} is a dynamic network that contains both undirected and directed connections between the nodes.
\end{definition}

In this paper, we consider the particular class of mixed networks that consists of undirected diffusively coupled networks with additional directed dynamic links between the nodes. The most natural origin of the directed links is a digital controller, while they can also originate from nonsymmetric or directed physical components, such a diodes and check valves. An example of such a mixed network is a physical network with additional digital controllers or nonsymmetric components. Fig.~\ref{fig:class} shows a graph of a mixed network, where a directed dynamics from node $w_2$ to node $w_1$ (in red) is added to the undirected network (in blue). 

\begin{remark}[Extension] \label{rem:extension}
  An extension of the considered class of mixed networks concerns multiple distinct sets of nodes that may contain additional directed dynamics and that are being interconnected through directed dynamics \cite{KIVITS2024t}. 
\end{remark}

\begin{figure}
  \centering
  \includegraphics[trim={0 3cm 0 0},clip,width=\linewidth]{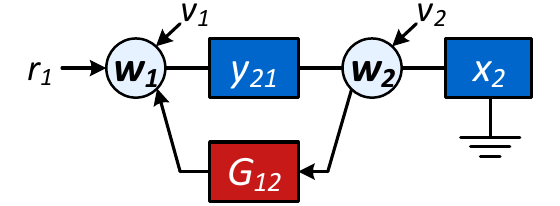}
  \caption{A mixed network with nodes $w_1$ and $w_2$ (blue circles); excitation $r_1$; disturbances $v_1$ and $v_2$; diffusive couplings (lines) with undirected dynamics (blue blocks); directed links (arrows) with input-output dynamics (red block); and a ground (triangle of three lines). }
  \label{fig:class}
\end{figure}

Including the directed dynamics in the diffusively coupled network model \eqref{eq:diffusive} leads to the mixed network model.
\begin{definition}[Mixed network model] \label{def:original}
  A mixed network model consists of $L$ node signals $w_j(t)$, $j=1,2,\ldots,L$, and $K$ excitation signals $r_k(t)$, $k=1,2,\ldots,K$, and is written as
  \begin{equation}\label{eq:model_original}
    A(\q) w(t) = B(\q) r(t) + G(q) w(t) + F(q) e(t),
  \end{equation}
  with $A(\q)\in\mathcal{A}$, $B(\q)\in\mathcal{B}$, $G(q)\in\mathcal{G}:=\{G\in\mathbb{R}^{L\times L}(q)\ |\ G_{jk}\ \mbox{proper},  G_{jj}=0,\ \forall j,k\}$, $F(q)\in\mathcal{F}$, and with $\Lambda$, $w(t)$, $r(t)$, and $e(t)$ as in Definition~\ref{def:diffusive}. This network is assumed to be well-posed and stable. 
\end{definition}

The mixed network is well-posed and stable if $(A(\q)-G(q))^{-1}$ exists and is proper and stable \cite{DANKERS2014}. The dynamics in the directed links are modelled by a rational matrix $G(q)\in \mathcal{G}$. Each element in $G(q)$ corresponds to an input-output system in the interconnection between two node signals. The zero structure of $G(q)$ represents the interconnection structure of the directed part of the network. Self-loops are excluded as, without loss of generality, $G(q)$ is hollow.

\begin{example}[Mixed network model] \label{ex:msd_mixed}
    The (continuous-time) model \eqref{eq:msd_mixed} of the controlled mass-spring-damper system described in Section~\ref{section:network} has the same structure as the (discrete-time) mixed network model in \eqref{eq:model_original} with $u(t):=B_0 r(t) + v(t)$ and with matrices $A_2=M_d$, $A_1=D_d+D_{\mathcal{L}}$, $A_0=K_d+K_{\mathcal{L}}$, and $G_0=G$. 
\end{example}

While the matrices $A(\q), B(\q)$ and $G(q)$ in \eqref{eq:model_original} are a mix of polynomial and rational terms, for algorithmic purposes it is attractive to rewrite the network equation into a form where the terms that process $w(t)$ and $r(t)$ all become polynomial. This is formulated next, where the transfer function matrix $G(q)$ is made polynomial by uniquely factorising it as
\begin{equation}\label{eq:model_G}
  G(q) = D_G^{-1}(\q) N_G(\q),
\end{equation}
with $D_G(\q)$ and $N_G(\q)$ being left coprime; $D_G(\q)$ diagonal, full rank, and monic; and the notational rule that $D_{G_{ii}}(\q) := 1$ if $G_{ij}(q) = 0$ for all $j$. 
\begin{proposition}[Polynomial mixed network model] \label{prop:poly}
  Any mixed network defined by Definition~\ref{def:original} can be written as
  \begin{equation}\label{eq:model_poly}
    \breve{\Upsilon}(\q) w(t) = \breve{B}(\q) r(t) + \breve{F}(q) e(t),
  \end{equation}
  with $\breve{\Upsilon}(\q) := \breve{A}(\q) - \breve{G}(\q)$ with polynomial elements $\breve{\Upsilon}_{ij}(\q) =: \breve{\upsilon}_{ij}(\q)$ and with
  \begin{subequations}\label{eq:model_breve}\begin{align}
    \breve{A}(\q) &= d_G(\q) A(\q),                   \label{eq:model_breveA}\\
    \breve{B}(\q) &= d_G(\q) B(\q),                   \label{eq:model_breveB}\\
    \breve{F}( q) &= d_G(\q) F(q),                    \label{eq:model_breveF}\\
    \breve{G}(\q) &= d_G(\q) D_G^{-1}(\q) N_G(\q),    \label{eq:model_breveG}\\
    d_G(\q) &:= \frac{det(D_G(\q))}{gcd(det(D_G(\q)),adj(D_G(\q)))},
  \end{align}\end{subequations}
   with $D_G(\q)$ and $N_G(\q)$ satisfying \eqref{eq:model_G}. Observe that $d_G(\q)$ is a monic and scalar polynomial. 
\end{proposition}
\begin{pf}
  Subtract $G(q)w(t)$ from both sides of the equation in \eqref{eq:model_original}, substitute \eqref{eq:model_G} for $G(q)$, and premultiply both sides of the equation with $d_G(\q)$.
\end{pf}

While $\breve{A}(\q)$ is symmetric, $\breve{G}(\q)$ is a nonsymmetric polynomial matrix destroying the symmetric nature of the network model. This means that, while the mixed network model \eqref{eq:model_original} has symmetric $A(\q)$, the mixed network model \eqref{eq:model_poly} has nonsymmetric $\breve{\Upsilon}(\q)$. In addition, the model matrices in \eqref{eq:model_original} are all pre-multiplied with monic polynomial $d_G(\q)$ to make the network model \eqref{eq:model_poly} polynomial. This means that the elements of $\breve{A}(\q)$, $\breve{B}(\q)$, and $\breve{F}(q)$ have the common term $d_G(\q)$. Further, $\breve{A}(\q)$, $\breve{B}(\q)$, $\breve{G}(\q)$ (and $N_G(\q)$), and $\breve{F}(q)$ adopt the zero structure of $A(\q)$, $B(\q)$, $G(q)$, and $F(q)$, respectively, meaning that they still capture the interconnection structure of the network. If the directed dynamics are described by a polynomial matrix $G(q):=G(\q)$ instead of a rational matrix, then $d_G(\q)=D_G(\q)=1$, simplifying \eqref{eq:model_poly} to $\Upsilon(\q) w(t) = B(\q) r(t) + F(q) e(t)$, with $\Upsilon(\q)=\bar{\Upsilon}(\q):= A(\q) - N_G(\q)$ being polynomial and nonsymmetric. For future reference, define $\Upsilon(q) := A(\q)-G(q)$. 

The input-output relations between the signals of the mixed network model are given by
\begin{equation} \label{eq:model_io}
  w(t) = T_{wr}(q) r(t) + \bar{v}(t), \quad \bar{v}(t) = T_{we}(q) e(t),
\end{equation}
with rational transfer function matrices
\begin{subequations} \label{eq:model_tf} \begin{align}
  T_{wr}(q) &= \breve{\Upsilon}^{-1}(\q) \breve{B}(\q) = \Upsilon^{-1}(q)B(\q), \label{eq:model_tfwr}\\
  T_{we}(q) &= \breve{\Upsilon}^{-1}(\q) \breve{F}(q) = \Upsilon^{-1}(q)F(q).
\end{align}\end{subequations}
and with spectral density
\begin{subequations} \label{eq:model_spectrum} \begin{align}
    \Phi_{\bar{v}}(\omega) :=& ~\Four\{\E[\bar{v}(t)\bar{v}^{\top}(t-\tau) ] \}, \label{eq:model_spectrum1}\\
    =& ~T_{we}(e^{i\omega}) \Lambda T^{\ast}_{we}(e^{i\omega}), \label{eq:model_spectrum2}
\end{align}\end{subequations}
with $\Four$ the discrete-time Fourier transform and $(\cdot)^{\ast}$ the complex conjugate transpose. 

\section{Identification} \label{section:identification}
The full network identification problem that will be considered in this paper concerns the identification of all dynamics in the mixed network. This means that the objective is to identify all polynomial and rational terms in \eqref{eq:model_original}. To do this, we follow a similar reasoning as for full network identification of diffusively coupled networks \cite{KIVITS2023_TAC}, with the main difference being the presence of $G(q)$, which is rational and not polynomial that has to be accounted for. The dynamics of the mixed network are estimated via a prediction-error approach in which all node signals $w(t)$ are predicted based on the measured signals that are available. The one-step-ahead predictor is defined as in \cite{KIVITS2023_TAC}:
\begin{equation}\label{eq:predictor_def}
  \hat{w}(t|t-1) := \E\{w(t) ~|~ w^{t-1}, ~r^t \},
\end{equation}
where $w^{\ell}$ and $r^{\ell}$ refer to the signal samples $w(\tau)$ and $r(\tau)$, respectively, for all $\tau\leq\ell$. For the mixed network model \eqref{eq:model_original}, this one-step-ahead predictor is given by
\begin{equation}\label{eq:predictor}
 \hat{w}(t|t-1) = W(q)z(t)
\end{equation}
with data $z^{\top}(t):=\begin{bmatrix}w^{\top}(t)&r^{\top}(t)\end{bmatrix}$ and with predictor filter $W(q) = \begin{bmatrix}1-W_w(q)&W_r(q)\end{bmatrix}$ with 
\begin{subequations}\label{eq:predictor_filters}\begin{align}
  W_w(q) &= (A_0 - G^{\infty})^{-1} F^{-1}(q) (A(\q) - G(q)), \\
  W_r(q) &= (A_0 - G^{\infty})^{-1} F^{-1}(q) B(\q),
\end{align}\end{subequations}
with $A_0 := \lim_{z\rightarrow\infty}A(z)$ and  $G^{\infty} := \lim_{z\rightarrow\infty} G(z)$.

The models that will be considered are gathered in the network model set $\mathcal{M} := \{ M(\theta), \theta\in\Theta\subset \mathbb{R}^d\}$ with particular models
\begin{equation}\label{eq:models}
  M(\theta) := ( A(\q,\theta), B(\q,\theta), F(q,\theta), G(q,\theta), \Lambda(\theta))
\end{equation}
satisfying the properties of \eqref{eq:model_original}, where $\theta$ contains all unknown coefficients that appear in the entries of the model matrices. This leads to a parameterised one-step-ahead predictor $\hat{w}(t|t-1;\theta) = W(q,\theta)z(t)$ and a related parameterised one-step-ahead prediction error $\varepsilon(t,\theta) := w(t) - \hat{w}(t|t-1;\theta)$. The parameters of the mixed network model are then estimated through the weighted least-squares identification criterion
\begin{equation}\label{eq:criterion_theta}
  \tn = \arg \min_{\theta\in\Theta} \frac{1}{N} \sum_{t=1}^{N} \varepsilon^{\top}(t,\theta) S \varepsilon(t,\theta)
\end{equation}
with a user-designed weighting matrix $S\succ0$.

In the sequel of this paper, we will address two main questions: (a) under which conditions does this estimate lead to a consistent estimate of the mixed network coefficients; and (b) is there a practical alternative for the algorithmic optimization of \eqref{eq:criterion_theta}, which is a nonstandard and not very attractive nonconvex optimization problem due to the combination of polynomial and rational terms.

For showing consistency of the estimator \eqref{eq:criterion_theta} we first need to investigate under which conditions the considered model set of mixed network models is identifiable. 

\section{Network identifiability} \label{section:identifiability}

\subsection{Introduction} \label{subsection:identifiability_intro}
For effectively using the model set $\mathcal{M}$ based on \eqref{eq:models} in the identification procedure, we need to make sure that the different network models in $\mathcal{M}$ are distinguishable on the basis of measurement data. This is captured in a uniqueness-oriented concept of network identifiability, the definition of which we adopt from \cite{WEERTS2018_AUTO1}.

\begin{definition}[Network identifiability \cite{WEERTS2018_AUTO1}.] \label{def:netidf}
  The network model set $\mathcal{M}$ is \emph{globally network identifiable} from data $z^{\top}(t):=\begin{bmatrix}w^{\top}(t)& r^{\top}(t)\end{bmatrix}$ if the parameterised model ${M}(\theta)$ can uniquely be recovered from $T_{wr}(q,\theta)$ and $\Phi_{\bar{v}}(\omega,\theta)$. That is, if for all models ${M}(\theta_1), {M}(\theta_2) \in {\mathcal{M}}$ it holds that
  \begin{equation} \label{eq:netidf} \begin{rcases*}
    T_{wr}(q,\theta_1) = T_{wr}(q,\theta_2) \\
    \Phi_{\bar{v}}(\omega,\theta_1) = \Phi_{\bar{v}}(\omega,\theta_2)
    \end{rcases*} \Rightarrow M(\theta_1) = M(\theta_2).
    \end{equation}
\end{definition}

Network identifiability for diffusively coupled networks has been studied in \cite{KIVITS2023_TAC} for networks that can be represented in polynomial form. The current mixed situation that includes rational terms asks for some extensions of this analysis. In order to benefit maximally from identifiability results for polynomial representations, we first consider a polynomial model class that is typically larger than the model set $\mathcal{M}$. In a second step we will then analyse identifiability of $\mathcal{M}$.

\subsection{Polynomial parameterisation} \label{subsection:identifiability_poly}
As a first step in analysing network identifiability for model set $\mathcal{M}$ \eqref{eq:models}, we will study a closely related model set $\breve{\mathcal{M}}$ corresponding to the polynomial mixed network model \eqref{eq:model_poly} as formulated in Proposition~\ref{prop:poly}. The main terms in this model set are polynomial and independently parameterised, where we actually discard the common term $d_G(\q)$ present in \eqref{eq:model_breveA}-\eqref{eq:model_breveF}.

\begin{definition}[Polynomial parameterisation] \label{def:breve}
  We define a network model set $\breve{\mathcal{M}} := \{ \breve{M}(\eta), \eta\in\Pi\subset \mathbb{R}^m\}$, through polynomial parameterisation
  \begin{equation}\label{eq:models_breve}
    \breve{M}(\eta) := ( \breve{A}(\q,\eta), \breve{B}(\q,\eta), \breve{F}(q,\eta), \breve{G}(\q,\eta), \breve{\Lambda}(\eta))
  \end{equation}
  where each model represents a network according to \eqref{eq:model_poly}:
  \begin{subequations}\begin{align}
    \breve{\Upsilon}(\q,\eta)w(t) &= \breve{B}(\q,\eta)r(t) + \breve{F}(q,\eta) e(t), \\ 
    \breve{\Upsilon}(\q,\eta) &= \breve{A}(\q,\eta)-\breve{G}(\q,\eta),
  \end{align}\end{subequations}
  with $\breve{A}(\q)$, $\breve{B}(\q)$, and $\breve{G}(\q)$ polynomial; $\breve{A}(\q)$ symmetric; $\breve{A}(\q)$, $\breve{B}(\q)$, $\breve{F}(q)$, and $\breve{G}(\q)$ adopting the zero-structure of the related terms $A(\q)$, $B(\q)$, $F(q)$, and $G(q)$, respectively; and where $\eta$ independently parameterises all unknown terms in these matrices.
\end{definition}

Note that the description of $\breve{M}(\eta)$ follows the polynomial representation of mixed networks as formulated in Proposition \ref{prop:poly}, but without the common term $d_G(\q)$ in $\breve{A}(\q)$, $\breve{B}(\q)$ and $\breve{F}(q)$. The constraint of having this common term is therefore relaxed in the model set $\breve{\mathcal{M}}$.

The resulting network identifiability conditions for the model set $\breve{\mathcal{M}}$ can now be formulated as an extension of the network identifiability result for diffusively coupled networks \cite{KIVITS2023_TAC}:
\begin{proposition}[Network identifiability of $\breve{\mathcal{M}}$] \label{prop:identifiability_breve}
  Network model set $\breve{\mathcal{M}}$ is \emph{globally network identifiable} from data $z(t)$ if the following conditions are satisfied:
  \begin{enumerate}
    \item $\breve{\Upsilon}(\q)$ and $\breve{B}(\q)$ are left coprime in $\breve{\mathcal{M}}$.
    \item A permutation matrix $P$ exists such that within $\breve{\mathcal{M}}$, $\begin{bmatrix}\breve{\Upsilon}_0 & \breve{\Upsilon}_1 & \cdots & \breve{\Upsilon}_{n_{\breve{\upsilon}}} & \breve{B}_0 & \breve{B}_1 & \cdots & \breve{B}_{n_{\breve{b}}}\end{bmatrix} P = \begin{bmatrix}\breve{D} & \breve{R} \end{bmatrix}$ with $\breve{D}$ square, diagonal, and full rank.
    \item At least one excitation signal $r_k(t)$ is present.
    \item There is at least one constraint on the parameters of $\breve{\Upsilon}(\q,\eta_{\Upsilon})$ and $\breve{B}(\q,\eta_B)$ of the form $\breve{\Gamma}\bar{\eta}=\breve{\gamma}\neq0$, with $\breve{\Gamma}$ full row rank and with $\bar{\eta}^{\top} := \begin{bmatrix}\eta_{\Upsilon}^{\top}&\eta_B^{\top}\end{bmatrix}$.
    \item The zero structure of $\breve{G}(\q)$ is known in $\breve{\mathcal{M}}$.
    \item $\breve{g}_{ji}(\q)=0$ if $\breve{g}_{ij}(\q)\neq0$, $\forall i,j$, within $\breve{\mathcal{M}}$.
  \end{enumerate}
\end{proposition}
\begin{pf}
  The proof is included in Appendix~\ref{app:prop_idfbreve}.
\end{pf}

Condition (3) is needed so that $T_{wr}(q,\eta)$ exists, because $T_{we}(q,\eta)$ is not sufficient for determining a unique $\breve{\Upsilon}(\q)$. The parameter constraint in Condition (4) means that one nonzero parameter or the sum of some nonzero parameters is known. Conditions (5) implies that the locations of the directed links are known, while Condition (6) indicates that there can only be a directed link in at most one direction between two nodes. Condition (1) ensures that there are no common terms introduced when $T_{wr}(q)$ is factorised into polynomial matrices. The diagonality constraint of Condition (2) is explained in Section~\ref{subsection:identifiability_original}, where network identifiability conditions are formulated for the original network model \eqref{eq:model_original} and additional interpretation of these conditions is provided. 

\subsection{Network identifiability of $\mathcal{M}$} \label{subsection:identifiability_original}
To translate the network identifiability conditions for the mixed network in polynomial form $\breve{\mathcal{M}}$ into the network identifiability conditions for the original representation $\mathcal{M}$, the following result is achieved:

\begin{proposition}[Network identifiability] \label{prop:identifiability}
  Network model set $\mathcal{M}$ is \emph{globally network identifiable} from data $z(t)$ if the following conditions are satisfied:
  \begin{enumerate}
    \item $\begin{bmatrix}A(\q) & N_G(\q)\end{bmatrix}$ and $B(\q)$ left coprime in $\mathcal{M}$.
    \item A permutation matrix $P$ exists such that within $\mathcal{M}$, $\begin{bmatrix}\bar{\Upsilon}_0 & \bar{\Upsilon}_1 & \cdots & \bar{\Upsilon}_{n_{\upsilon}} & B_0 & B_1 & \cdots & B_{n_b}\end{bmatrix} P = \begin{bmatrix}D & R\end{bmatrix}$ with $D$ square, diagonal, and full rank.
    \item At least one excitation signal $r_k(t)$ is present.
    \item There is at least one constraint on the parameters of $A(\q,\theta_A)$ and $B(\q,\theta_B)$ of the form $\Gamma\bar{\theta}=\gamma\neq0$, with $\Gamma$ full row rank and with $\bar{\theta}^{\top} := \begin{bmatrix}\theta_A^{\top}&\theta_B^{\top}\end{bmatrix}$.
    \item The zero structure of $G(q)$ is known in $\mathcal{M}$.
    \item $G_{ji}(q)=0$ if $G_{ij}(q)\neq0$, $\forall i,j$, within $\mathcal{M}$.
  \end{enumerate}
\end{proposition}
\begin{pf}
  The proof is included in Appendix~\ref{app:prop_idf}.
\end{pf}

The conditions have a similar interpretation as the conditions in Proposition~\ref{prop:identifiability_breve}. To elaborate on this, the parameter constraint of Condition (4) is, for example, satisfied if one interconnection is (partially) known (in $A(\q,\theta_A)$) or if one excitation signal enters a node through (partially) known dynamics (in $B(\q,\theta_B)$). If the directed dynamics are polynomial, then $d_G(\q) = 1$ and the parameters of $G(\q,\theta_G)$ can be included in this parameter constraint. The constraint is then also satisfied if the directed dynamics in one link are (partially) known. Condition (5) means that the locations of the directed links are known. Condition (6) indicates that $G_{ji}(q)$ and $G_{ij}(q)$ cannot be nonzero at the same time, meaning that there can only be a directed link in at most one direction between two nodes. As there are no requirements on the zero-structure of $A(q)$, the topology of the undirected part can still be unknown. The diagonality constraint of Condition (2) can be understood through the following example: 

\begin{example}[Diagonality constraint] \label{ex:msd_diag}
    Consider the mass-spring-damper system of Example~\ref{ex:msd_diffusive} and consider the following situations: (1) if $B_0$ is square, diagonal and full rank, then all nodes have their own external input signal; (2) if $A_2$ is square, diagonal and full rank, then all nodes have a nonzero mass; (3) if $A_1$ is square, diagonal and full rank, then a damper is present in the connection of each node to the ground and no damper is present in the connections between the nodes; (4) if $A_0$ is square, diagonal and full rank, then a spring is present in the connection of each node to the ground and no spring is present in the connections between the nodes. In each of these cases, the columns of $\begin{bmatrix}A_0&A_1&A_2&B_0\end{bmatrix}$ can easily be sorted such that the first three columns form a square, diagonal, full rank matrix. This can also be done if for each node $w_j$, $j=1,2,3$, one of the above situations holds. For example, if node 1 has a nonzero mass ($M_{10}\neq0$); node 2 has an external excitation signal that is not entering any other node (e.g. $B_0^{\top}=\begin{bmatrix}0&1&0\end{bmatrix}$); and node 3 has a damper connecting it to the ground and no damper in the connections with other nodes ($D_{30}\neq0$,$D_{3k}=0$ for $k\neq0$). 
\end{example}

\section{Consistency of $M(\tn$)} \label{section:consistency} 
To show consistency of the network estimator $M(\tn)$ for the parameter estimator \eqref{eq:criterion_theta}, the data need to contain sufficient information to uniquely recover the predictor filter from second-order statistical properties of the data. This concept of data informativity is formalised in \cite{LJUNG1999}.
\begin{definition}[Data informativity \cite{LJUNG1999}] \label{def:datainf}
  A quasi-stationary data sequence $\{z(t)\}$ is called \emph{informative with respect to the model set $\mathcal{M}$} if for any two $\theta_1,\theta_2\in\Theta$ and predictor filter $W(q,\theta)$,
   \begin{multline}\label{eq:datainf}
    \Eb\left\{\| [W(q,\theta_1)-W(q,\theta_2)]z(t)\|^2 \right\}=0 \\ \Rightarrow \quad \{W(e^{i\omega},\theta_1) = W(e^{i\omega},\theta_2)\}
  \end{multline}
  for almost all $\omega$. 
\end{definition}

Applying Definition~\ref{def:datainf} to mixed networks, leads to the condition for data informativity.
\begin{proposition}[Data informativity] \label{prop:datainformativity}
  The quasi-stationary data sequence $\{z(t)\}$ is \emph{informative with respect to the model set $\mathcal{M}$} if $\Phi_z(\omega)\succ 0$ for a sufficiently large number of frequencies. In the situation $K \geq 1$, this is guaranteed by $\Phi_r(\omega)\succ 0$ for a sufficiently large number of frequencies.
\end{proposition}
\begin{pf}
  The proof is included in Appendix~\ref{app:prop_datainf}. 
\end{pf}

The exact number of frequencies that is sufficient is determined by the number of parameters that are present in the model. The condition that $\Phi_r(\omega)\succ0$ for a sufficiently large number of frequencies seems to be a general condition. However, the dimension of $\Phi_r(\omega)$ depends on the number of excitation signals $r_j(t)$, denoted by $K$, which is specified in the model set. Thus, all excitations signals that are present (according to the model set), need to be persistently exciting. This is because each additional excitation signal $r_j(t)$ also introduces new polynomials $b_{kj}(\q)$ that need to be identified. Further, the data-informativity condition in Proposition~\ref{prop:datainformativity} is equal to that for diffusively coupled networks \cite[Proposition~4]{KIVITS2023_TAC}, which implies that the directed links can be identified for free. This is because the mixed network is a diffusively coupled network with no additional node signals, but with only (directed) couplings added to it.

Informativity of $z(t)$ implies that $W(q)$ can uniquely be recovered from data, which by \eqref{eq:predictor_filters} and \eqref{eq:model_io}-\eqref{eq:model_tf} implies that the pair $(T_{wr}(q),\Phi_{\bar v}(\omega))$ can uniquely be recovered. Identifiability of $\mathcal{M}$ implies that $A(\q,\theta)$, $B(\q,\theta)$, $F(q,\eta)$, $G(q,\theta)$, and $\Lambda(\theta)$ can uniquely be determined from the pair $(T_{wr}(q),\Phi_{\bar v}(\omega))$. If the model set contains the true system, then the conditions for data informativity and network identifiability together induce the consistency result.

Let the data that are available for identification be generated by the system $\mathcal{S} := (A^0, B^0, F^0, G^0, \Lambda^0)$. This true system lies in the model set if $\mathcal{S} = M(\theta^0)$, with $\theta^0\in\Theta$ indicating the true parameter values. 

\begin{theorem}[Consistency] \label{thm:consistency} 
  Consider a data generating system $\mathcal{S}$ and a model set $\mathcal{M}$ with uniformly stable predictor filter $W(q,\theta)$ satisfying \eqref{eq:predictor_filters}. Then $M(\tn)$ is a consistent estimate of $\mathcal{S}$ if the following conditions hold:
  \begin{enumerate}
    \item The true system lies in the model set.
    \item The data are informative with respect to the model set.
    \item The model set is globally network identifiable.
  \end{enumerate}
\end{theorem}
\begin{pf}
  The proof is included in Appendix~\ref{app:thm_consistency}. 
\end{pf}

Any weight $S\succ0$ leads to consistent estimates, but the minimum variance is only achieved for $S=(\Lambda^0)^{-1}$. 

Now that it has been proved that mixed networks \eqref{eq:model_original} can consistently be identified, the next step is to formulate an algorithm for obtaining the estimate.

\section{Network identification algorithm} \label{section:algorithm}

\subsection{Identification approach} \label{subsection:approach}
The most straightforward choice for estimating $\tn$ would be to directly solve the estimation problem \eqref{eq:criterion_theta}. However, given the combination of polynomial and rational terms in the models $M(\theta)$ \eqref{eq:models}, this will be a nonstandard and nonconvex optimization problem. This is not attractive, in particular when the dimension of the network (i.e. the number of node signals in $w(t)$) is large. In line with the algorithm presented in \cite{KIVITS2023_TAC} for diffusively coupled networks, we present an alternative two-step approach: 
\begin{enumerate}
    \item Estimate a polynomial model in the model set $\breve{\mathcal{M}}$ as formulated in Definition~\ref{def:breve}.
    \item Map the estimated parameters to a model in the model set $\mathcal{M}$, thus obtaining the physical coefficients of the mixed network.
\end{enumerate}
These two steps are subsequentially discussed. 

\subsection{Step 1: Estimation of the polynomial form} \label{subsection:estimation_poly}
For estimating the polynomial form, the multi-step algorithm for diffusively coupled networks \cite{KIVITS2023_TAC} can be adapted to account for the directed dynamics that are present in the mixed network. 

The polynomial model set $\breve{\mathcal{M}}$ can be identified with a prediction-error approach similar to the one formulated in Section~\ref{section:identification} for the original model set $\mathcal{M}$. Consistency of the network estimator $\breve{M}(\en)$ can be proved along the same lines as consistency of $M(\tn)$ as described in Section~\ref{section:consistency}. The conditions for network identifiability of $\breve{\mathcal{M}}$ are formulated in Proposition~\ref{prop:identifiability_breve}, while data informativity directly follows from Proposition~\ref{prop:datainformativity}, because both models $\breve{M}(\eta)$ and $M(\theta)$ describe the same input-output behaviour \eqref{eq:model_io}-\eqref{eq:model_tf}. The final consistency result is closely related to the one for diffusively coupled networks \cite{KIVITS2023_TAC}, with the addition that Condition (5) and (6) in Proposition~\ref{prop:identifiability_breve} for network identifiability have to be satisfied to account for the directed dynamics represented by $\breve{G}(\q)$ that also make $\breve{\Upsilon}(\q)$ partially nonsymmetric. 

This analysis leads to a nonconvex optimization problem, which is not attractive as an algorithm. Especially for networks with many node signals in $w(t)$, this results in high computational complexity and occurrences of local optima. Alternatively, the multi-step algorithm formulated in \cite{KIVITS2023_TAC} is proposed, because this algorithm is convex and leads to a unique solution. The polynomial form~\eqref{eq:model_poly} can consistently be identified by incorporating the partially asymmetric structure of $\breve{\Upsilon}(\q)$ (defined by the known structure of $\breve{G}(\q)$) into the multi-step algorithm for diffusively coupled networks \cite{KIVITS2023_TAC}\footnote{This algorithm is formulated for the situation where $F(q)=F(\q)$ is polynomial. However, the algorithm can be adapted to handle the situation where $F(q)$ is rational.}. This algorithm consists of four steps: 
\begin{enumerate}[label=(\alph*)]
    \item Estimate a nonparametric autoregressive with exogenous variables (ARX) model by least-square. 
    \item Reduce the nonparametric ARX model to a parametric model by weighted constrained least-squares. 
    \item Improve the parametric model by weighted constrained least-squares.
    \item Obtain the noise model by direct calculation.
\end{enumerate}

Step (a) of this algorithm does not include any structural information of the model, while an infinite series expansion is parameterised by high-order polynomials. This high-order ARX model is used to identify a structured network model in step (b), by incorporating the symmetric and nonmonic structure of $\breve{A}(\q)$ and the asymmetric structure of $\breve{G}(\q)$ in the parameterisation. The constraint comes from the parameter constraint as described by Condition (4) of Proposition~\ref{prop:identifiability_breve}. Step (c) aims to correct for the residuals that are not accounted for in step (a), due to the high-order approximation of the nonparametric ARX model. Step (a)-(c) are the main three steps of the algorithm and all include a convex least-squares optimization. Finally, the noise model is directly calculated from the parametric model to correct the noise model for the nonmonicity of $\breve{A}(\q)$. 

\subsection{Step 2: Estimation of the network model} \label{subsection:estimation_original}

\subsubsection{Identification approach}
After estimating a polynomial model in the model set $\breve{\mathcal{M}}$, the estimated parameters $\en$ can be mapped into a model in the model set $\mathcal{M}$, with which we obtain the physical coefficients of the mixed network. We do this by exploiting the mapping in \eqref{eq:model_breveA}-\eqref{eq:model_breveG} and by showing that a consistent estimate of $\eta$ leads to a consistent estimate of $\theta$.

According to Lemma~\ref{lem:unique} in Appendix~\ref{app:prop_idf}, $A(\q)$, $B(\q)$, and $G(q)$ can uniquely be found from $\breve{A}(\q)$, $\breve{B}(\q)$, and $\breve{G}(\q)$ if Condition (4) of Proposition~\ref{prop:identifiability} is satisfied. Algorithmically, this can be solved using three so-called (weighted) null-space fitting (WNSF) steps \cite{GALRINHO2019}, which are convex least-square optimization steps. Finally, the noise model is recovered, which follows from direct calculation. This gives the following algorithmic steps: 
\begin{enumerate}[label=(\alph*)]
    \item Estimate $A(\q)$ and $B(\q)$ by exploiting \eqref{eq:model_breveA} and \eqref{eq:model_breveB}.
    \item Estimate the common scalar polynomial $d_G(\q)$ from the estimated $A(\q)$ and $B(\q)$ and the relations \eqref{eq:model_breveA} and \eqref{eq:model_breveB}.
    \item Estimate the directed dynamics $G(q)$ from the estimated $d_G(\q)$ and the relations \eqref{eq:model_breveG} and \eqref{eq:model_G}. 
    \item Obtain the noise model represented by $F(q)$ and $\Lambda$ from the estimated $d_G(\q)$ and \eqref{eq:model_breveF}. 
\end{enumerate}
These steps are subsequently discussed, after which we show that this leads to a consistent estimate of $\theta$.

\subsubsection*{(a) Estimate $A(\q)$ and $B(\q)$}
Consider the model structure $\breve{A}(\q,\eta_{\breve{A}})$ and $\breve{B}(\q,\eta_{\breve{B}})$ as models of $\breve{A}(\q)$ and $\breve{B}(\q)$, respectively, and the model structure $A(\q,\theta_A)$ and $B(\q,\theta_B)$ as models of $A(\q)$ and $B(\q)$, respectively. For notational elegance, let us define the parameter vectors $\zeta^{\top}:=\begin{bmatrix}\eta_{\breve{A}}^{\top}&\eta_{\breve{B}}^{\top}\end{bmatrix}$ and $\vartheta^{\top}:=\begin{bmatrix}\theta_A^{\top}&\theta_B^{\top}\end{bmatrix}$.

As $d_G(\q)$ is a scalar polynomial and $A(\q)$ is invertible, \eqref{eq:model_breveA} can be used to express $d_G^0(\q)$ as 
\begin{equation}\label{eq:relation_dA}
  d_G^0(\q) I = \big(A^0(\q)\big)^{-1} \breve{A}^0(\q),
\end{equation}
where $I$ is the identity matrix of appropriate size. Together with \eqref{eq:model_breveB} this gives
\begin{equation}\label{eq:relation_AB}
  A^0(\q) \breve{B}^0(\q) - \breve{A}^0(\q) B^0(\q) = 0,
\end{equation}
from which we can extract 
\begin{equation}\label{eq:estimate_nullspace_AB}
    -Q(\zeta^0)\vartheta^0=0,
\end{equation}
where $\vartheta^0$ represents the coefficients of $A^0(\q)$ and $B^0(\q)$ and where $\breve{A}^0(\q)$ and $\breve{B}^0(\q)$ of the polynomial representation are incorporated in $Q(\zeta^0)$. On the basis of the estimated parameters of the polynomial form $\zn$ as part of $\en$, an estimate of $\vartheta^0$ is obtained by the linear constrained optimization problem
\begin{subequations}\label{eq:estimate_vartheta_problem}\begin{align}
    \vn =& \min_\theta \vartheta^{\top} Q^{\top}(\zn) Q(\zn) \vartheta\\
    &\text{subject to} \quad \Gamma\vartheta=\gamma \label{eq:estimate_optimization_constraint}
\end{align}\end{subequations}
where the constraint \eqref{eq:estimate_optimization_constraint} results from the parameter constraint for identifiability as formulated in Condition (4) of Proposition~\ref{prop:identifiability}. The optimization problem \eqref{eq:estimate_vartheta_problem} can be solved using the Lagrangian and the Karush-Kuhn-Tucker conditions, resulting in
\begin{equation}\label{eq:estimate_vartheta}
    \begin{bmatrix}\vn\\ \hat{\lambda}_N\end{bmatrix} = \begin{bmatrix}Q^{\top}(\zn) Q(\zn)&\Gamma^{\top}\\ \Gamma&0\end{bmatrix}^{-1} \begin{bmatrix}0\\ \gamma\end{bmatrix},
\end{equation}
where $\hat{\lambda}_N$ are the estimated Lagrange multipliers.

\subsubsection*{(b) Estimate the common polynomial term}
In this step, we estimate the scalar polynomial $d_G^0(\q)$, which is the common term in $\breve{A}^0(\q)$ and $\breve{B}^0(\q)$. From \eqref{eq:model_breveA} and \eqref{eq:model_breveB} we obtain 
\begin{subequations}\label{eq:relations_dAB}\begin{align}
    \breve{A}^0(\q) - d_G^0(\q) A^0(\q) = 0,\\
    \breve{B}^0(\q) - d_G^0(\q) B^0(\q) = 0,
\end{align}\end{subequations}
from which we can extract
\begin{equation}\label{eq:estimate_nullspace_d}
    p(\vartheta^0,\zeta^0) - P(\vartheta^0)\beta^0 = 0,
\end{equation}
where $\beta^0$ represents the coefficients of $d_G^0(\q)$ and where $A^0(\q)$, $B^0(\q)$, $\breve{A}^0(\q)$, and $\breve{B}^0(\q)$ are incorporated in $p(\vartheta^0,\zeta^0)$ and $P(\vartheta^0)$. The structural difference with \eqref{eq:estimate_nullspace_AB} is due to the monicity of $d_G^0(\q)$. The fact that $d_G^0(\q)$ is monic is directly incorporated in the parameterisation to avoid any constraints in the optimization problem. On the basis of the estimated parameters $\zn$ and $\vn$, an estimate of $\beta^0$ is obtained by linear least-squares, resulting in 
\begin{equation}\label{eq:estimate_beta}
    \bn = \left[ P^{\top}(\vn) P(\vn) \right]^{-1} P^{\top}(\vn) p(\vn,\zn).
\end{equation}

\subsubsection*{(c) Estimate the directed dynamics}
The directed dynamics, represented by $G^0(q)$ and factorised as in \eqref{eq:model_G}, are estimated from \eqref{eq:model_breveG} and $d_G(\bn)$. Let $D_G(\q,\theta_{D_G})$ and $N_G(\q,\theta_{N_G})$ be models of $D_G^0(\q)$ and $N_G^0(\q)$, respectively, and for notational elegance define $\theta_G^{\top}:=\begin{bmatrix}\theta_{D_G}^{\top}&\theta_{N_G}^{\top}\end{bmatrix}$. From \eqref{eq:model_breveG} we obtain
\begin{equation}\label{eq:relations_dG}
    D_G^0(\q) \breve{G}^0(\q) - d_G^0(\q) N_G^0(\q) = 0,
\end{equation}
where we use that $d^0_G(\q)$ is scalar and $D^0_G(\q)$ is invertible. From \eqref{eq:relations_dG} we can extract
\begin{equation} \label{eq:estimate_nullspace_G} 
    r(\eta^0) - R(\eta^0,\beta^0)\theta_G^0 = 0,
\end{equation}
where $\theta_G^0$ represents the coefficients of $D_G^0(\q)$ and $N_G(\q)$, and where $\breve{G}^0(\q)$ and $d_G^0(\q)$ are incorporated in $r(\eta^0)$ and $R(\eta^0,\beta^0)$. Similar to Step (b), an estimate of $\theta_G^0$ is obtained by linear least-squares as 
\begin{equation}\label{eq:estimate_thetaG}
    \hat{\theta}_{G_N} = \left[ R^{\top}(\en,\bn) R(\en,\bn) \right]^{-1} R^{\top}(\en,\bn) r(\en).
\end{equation}

\subsubsection*{(d) Obtain the noise model}
Having estimated $\breve{F}^0(q)$, $\breve{\Lambda}^0$, and $d_G^0(\q)$, the noise model represented by $F^0(q)$ and $\Lambda^0$ can be recovered. On the basis of \eqref{eq:model_breveF}, the estimate of $F^0(q)$ is constructed as
\begin{equation}\label{eq:estimate_F}
    F(q,\tn) = d_G^{-1}(\q,\bn) \breve{F}(q,\en). 
\end{equation}
Furthermore, as $\breve{\Lambda}^0=\Lambda^0$, the estimate of $\Lambda^0$ is simply 
\begin{equation}\label{eq:estimate_Lambda}
    \Lambda(\tn)=\breve{\Lambda}(\en). 
\end{equation}

\begin{remark}[Polynomial noise model]
    If $F(q)=F(\q)$ is polynomial, then \eqref{eq:model_breveF} can be incorporated in Step (a) and (b) of the algorithm above. 
\end{remark}

\subsubsection{Consistency of the result} \label{subsection:consistency_original}
Step (a)-(c) only linear optimization problems are solved and therefore, these optimizations are convex and have unique solutions. Step (d) is a direct calculation and therefore also unique. This means that any $\en$ obtained in Step 1 leads to a unique $\tn$ at the end of Step 2. Now it remains to show that if Step 1 results in $\en=\eta^0$, the result of Step 2 is exactly $\tn=\theta^0$, in other words, that Step 2 maps $\eta^0$ into $\theta^0$. 
\begin{proposition}\label{prop:consistency_alg}
    For every network model $M(\theta^0)$, there exists a polynomial model $\breve{M}(\eta^0)$, such that
    \begin{enumerate}[label=(\alph*)]
        \item \leavevmode\vspace{-\baselineskip} \begin{subequations}\label{eq:equal}\begin{align}
                     T_{wr}(q,\eta^0) &= T_{wr}(q,\theta^0), \\
        \Phi_{\bar{v}}(\omega,\eta^0) &= \Phi_{\bar{v}}(\omega,\theta^0).
        \end{align}\end{subequations}
        \item $M(\theta^0)$ is obtained from $\breve{M}(\eta^0)$ by solving for the linear equations \eqref{eq:estimate_nullspace_AB}, \eqref{eq:estimate_nullspace_d}, \eqref{eq:estimate_nullspace_G}, $\breve{F}(q,\eta^0)=d_G(\q,\beta^0)F(q,\theta^0)$, and $\breve{\Lambda}(\eta^0)=\Lambda(\theta^0)$, while uniqueness of $M(\theta^0)$ is guaranteed for a network identifiable model set $\mathcal{M}$ that contains $M(\theta^0)$.
    \end{enumerate}
\end{proposition}
\begin{pf}
    Substituting \eqref{eq:model_breve} into $T_{wr}(q,\eta^0)$ and $T_{we}(q,\eta^0)$ gives $T_{wr}(q,\eta^0)=T_{wr}(q,\theta^0)$, $T_{we}(q,\eta^0)=T_{we}(q,\theta^0)$, respectively, see \eqref{eq:model_tf}. Together with $\breve{\Lambda}(\eta^0)=\Lambda(\theta^0)$ and \eqref{eq:model_spectrum2} this also leads to $\Phi_{\bar{v}}(\omega,\eta^0) = \Phi_{\bar{v}}(\omega,\theta^0)$. Implication (b) follows directly from solving $M(\theta^0) = \breve{M}(\eta^0)$. 
\end{pf}

The linear estimation results from Step 2, given by \eqref{eq:estimate_vartheta}, \eqref{eq:estimate_beta}, \eqref{eq:estimate_thetaG}, \eqref{eq:estimate_F}, and  \eqref{eq:estimate_Lambda}, constitute a continuous mapping  $f: \breve{M}(\en) \rightarrow M(\tn)$. In case of a consistent estimate $\breve{M}(\en)$, it follows that $\breve{M}(\en) \underset{N\rightarrow\infty}{\longrightarrow} \breve{M}(\eta^0)$ with probability $1$\footnote{In accordance with \cite{LJUNG1999} this should be interpreted as that, for all polynomial/rational elements $X(z,\en)$ in the model $\breve{M}(\en)$, $\lim_{N\rightarrow\infty} X(\eio,\en) = X(\eio,\eta^0)$ with probability $1$ for almost all $\omega \in [0,\pi)$.}. From Proposition \ref{prop:consistency_alg} it follows that $f(\breve{M}(\eta^0)) = M(\theta^0)$, which implies that $M(\tn) \underset{N\rightarrow\infty}{\longrightarrow} M(\theta^0)$ with probability $1$, and thus the estimator $M(\tn)$ is consistent.

\subsection{Algorithm} \label{subsection:algorithm}
The above-mentioned steps describe the procedure for identifying the coefficients of a mixed network \eqref{eq:model_original}. This procedure is summarised in the following algorithm: 
\begin{algorit}[Mixed network] \label{alg:algorithm}
    Consider a data generating system $\mathcal{S}$ and a network model set $\mathcal{M}$. Then, $M(\tn)$, a consistent estimate of $M(\theta^0)$ is obtained through the following steps: 
    \begin{enumerate}
        \item Estimate a polynomial model in the model set $\breve{\mathcal{M}}$:
        \begin{enumerate}[label=(\alph*)]
            \item Estimate a nonparametric ARX model by least-square. 
            \item Reduce the nonparametric ARX model to a parametric model by weighted constrained least-squares. 
            \item Improve the parametric model by weighted constrained least-squares to obtain $\breve{A}(\q,\en)$, $\breve{B}(\q,\en)$, and $\breve{G}(\q,\en)$.
            \item Obtain the noise model by direct calculation to obtain $\breve{F}(q,\en)$ and $\breve{\Lambda}(\en)$.
        \end{enumerate}
        \item Map the estimated parameters $\en$ to a model in the model set $\mathcal{M}$:
        \begin{enumerate}[label=(\alph*)]
            \item Estimate $A(\q)$ and $B(\q)$ by constrained least-squares \eqref{eq:estimate_vartheta} to obtain $\vn$, resulting in $A(\q,\tn)$ and $B(\q,\tn)$.
            \item Estimate the common scalar polynomial by least-squares \eqref{eq:estimate_beta} to obtain $\bn$.
            \item Estimate the directed dynamics by least-squares by \eqref{eq:estimate_thetaG} to obtain $\hat{\theta}_{G_N}$, resulting in $G(\q,\tn)$. 
            \item Obtain the noise model by calculating \eqref{eq:estimate_F} and \eqref{eq:estimate_Lambda} resulting in $F(q,\tn)$ and $\Lambda(\tn)$.  
        \end{enumerate}
    \end{enumerate}
\end{algorit}

Consistency of the estimates $\en$ in Step 1 of Algorithm~\ref{alg:algorithm} follows from \cite{KIVITS2019} and is based on the similarity with WNSF and its proof \cite{GALRINHO2019}. Consistency of the estimates $\tn$ in Step 2 of Algorithm~\ref{alg:algorithm} follows from Proposition~\ref{prop:consistency_alg}, as discussed in Section~\ref{subsection:consistency_original}. 

\begin{remark}[Parameter constraint]
    The original mixed network model is parameterised in terms of $\theta$. Therefore, it is natural to have a parameter constraint on $\theta$, which is required for identifiability of $\mathcal{M}$ (see Condition (4) of Proposition~\ref{prop:identifiability}). However, this parameter constraint does in general not directly translate to a parameter constraint on $\eta$ as required for identifiability of $\breve{\mathcal{M}}$ (see Condition (4) of Proposition~\ref{prop:identifiability_breve}). A possible solution to this is to estimate a polynomial model with a user-defined parameter constraint (e.g. fixing one of the $\eta$-parameters to $1$) in step 1 of the algorithm. This results in a scaled version of the true parameters: $\eta^{\ast}=\alpha\eta^0$ with $\alpha\in\mathbb{R}_+$. The scaled polynomial model can still be used to find a consistent estimate of the original model in step 2, because the scaling of the parameters is fixed with the parameter constraint on $\theta$, see Lemma~\ref{lem:unique} in Appendix~\ref{app:prop_idf}. 
\end{remark}

\section{Discussion} \label{section:discussion}
For identifiability, only one directed interconnection between a pair of nodes is allowed, see Condition (6) of Proposition~\ref{prop:identifiability}. However, this restriction can be relaxed by replacing it with the condition that: for each $G_{ij}(q)\neq0$, $\forall i,j\in[1,L]$, either one of the following conditions is satisfied within $\mathcal{M}$:
\begin{enumerate}
  \item $G_{ji}(q)=0$.
  \item $x_{ii}(\q)=0$ and $G_{ik}(q)=0$ for all $k\neq j$ for which $G_{ki}(q)\neq0$. 
  \item $x_{jj}(\q)=0$ and $G_{kj}(q)=0$ for all $k\neq i$ for which $G_{jk}(q)\neq0$.
\end{enumerate}
Condition (1) is just Condition (6) of Proposition~\ref{prop:identifiability}. Condition (2) means that $w_i(t)$ has no connection to the ground and that from all $w_k(t)$ that have an incoming directed link from $w_i(t)$, $w_j(t)$ is the only one that also has an outgoing directed link towards $w_i(t)$. The idea is that the dynamics of all undirected links at $w_i(t)$ and all incoming links to $w_i(t)$ can be identified as before, except for the ones connected with $w_j(t)$. Then $a_{ii}(\q)$ is used to find $a_{ij}(\q)$, which leads to $G_{ij}(q)$. Condition (3) is dual to Condition (2). 

The undirected part of the mixed network represents a connected graph, meaning that if all directed links are removed from the network, each pair of nodes still has a path connecting them. In the extended case briefly mentioned in Remark~\ref{rem:extension}, some nodes are only interconnected through a directed interconnection, meaning that if all directed links are removed from the network, multiple distinct sets of nodes (i.e. diffusively coupled networks) remain. For these mixed networks, the identifiability conditions of Proposition~\ref{prop:identifiability} slightly change. Namely, instead of a single parameter constraint, there has to be a single parameter constraint for each distinct diffusively coupled network. In addition, an external excitation signal has to be present in each distinct diffusively coupled network.

A second class of mixed networks that can be considered consists of an undirected diffusively coupled network and a directed dynamic network that are interconnected with each other through directed links between the nodes. These mixed networks consist of a directed part and an undirected part, such as a network of communicating multi-agents where each each agent is controlled and the controllers exchange information among each other, or a physical network of which a part contains only directed components. An extension of this type of mixed networks consists of multiple distinct diffusively coupled networks sets and multiple distinct directed networks that are interconnected with each other. When combining these mixed networks with (the extended version of) the ones discussed in this article, any other type of mixed network can be constructed. 

\section{Conclusions} \label{section:conclusion}
Mixed networks are linear networks that contain both undirected and directed interconnections between the node signals. The particular class of mixed networks that has been considered consists of an undirected network with some additional directed links. The model of the mixed network has been reformulated into the polynomial framework. Conditions for identifiability and consistent parameter estimation of all dynamics in the network have been derived. A convex multi-step algorithm has been formulated to obtain a consistent estimate of the mixed network model. More advanced mixed networks are left for future research. 

\appendix
\section{Proof of Proposition~\ref{prop:identifiability_breve}} \label{app:prop_idfbreve}
The conditions for network identifiability are derived by using the left matrix-fraction description (LMFD) to recover $\breve{A}(\q)$, $\breve{B}(\q)$, and $\breve{G}(\q)$ in \eqref{eq:model_poly} from $T_{wr}(q)$. In addition, a spectral factorization is used to recover $\breve{F}(q)$ and $\breve{\Lambda}$ from $\Phi_{\bar{v}}(\omega)$.

\begin{lemma}[LMFD] \label{lem:LMFD}
  Consider a network model set $\breve{\mathcal{M}}$. Given the LMFD $T_{wr}(q)=\breve{\Upsilon}^{-1}(\q) \breve{B}(\q)$ \eqref{eq:model_tfwr}, the polynomial matrices $\breve{A}(\q)$, $\breve{G}(\q)$, and $\breve{B}(\q)$ are unique within $\breve{\mathcal{M}}$ up to a scalar factor if the following conditions hold:
  \begin{enumerate}
    \item $\breve{\Upsilon}(\q)$ and $\breve{B}(\q)$ are left coprime in $\breve{\mathcal{M}}$.
    \item A permutation matrix $P$ exists such that within $\breve{\mathcal{M}}$, $\begin{bmatrix}\breve{\Upsilon}_0 & \breve{\Upsilon}_1 & \cdots & \breve{\Upsilon}_{n_{\breve{\upsilon}}} & \breve{B}_0 & \breve{B}_1 & \cdots & \breve{B}_{n_{\breve{b}}}\end{bmatrix} P = \begin{bmatrix}\breve{D} & \breve{R} \end{bmatrix}$ with $\breve{D}$ square, diagonal, and full rank.
    \item The zero structure of $\breve{G}(\q)$ is known in $\breve{\mathcal{M}}$.
    \item $\breve{g}_{ji}(\q)=0$ if $\breve{g}_{ij}(\q)\neq0$, $\forall i,j$, within $\breve{\mathcal{M}}$.
  \end{enumerate}
\end{lemma}
\begin{pf}
  The LMFD of any two polynomial and left-coprime matrices is unique up to a premultiplication with a unimodular matrix \cite{KAILATH1980}. Let this premultiplication matrix be denoted by $U(q)$. The unimodular matrix $U(q)$ is restricted to be diagonal to preserve the diagonality of $\breve{D}$ in Condition (2). Let $\bar{A}(\q)$ denote the symmetric part of $\breve{\Upsilon}(\q)$. To preserve the symmetry of $\bar{A}(\q)$ (which has full rank), the diagonal premultiplication matrix $U(q)$ is further restricted to have equal elements. In other words, the premultiplication matrix $U(q)$ becomes an identity matrix that is scaled with a scalar factor. This means that $\breve{\Upsilon}(\q)$ and $\breve{B}(\q)$ are unique up to a scalar factor. Next, $\breve{\Upsilon}(\q)$ can be separated into $\breve{A}(\q)$ and $\breve{G}(\q)$ as follows: The symmetry of $\breve{A}(\q)$ together with Condition (3) and (4) give that $\breve{A}(\q)$ is found from $\breve{\Upsilon}(\q)$. Finally, $\breve{G}(\q) = \breve{A}(\q) - \breve{\Upsilon}(\q)$ and thus, $\breve{A}(\q)$ and $\breve{G}(\q)$ can uniquely be determined from $\breve{\Upsilon}(\q)$. As $\breve{\Upsilon}(\q)$ is unique up to a scalar factor, so are $\breve{A}(\q)$ and $\breve{G}(\q)$.
\end{pf}

Using the result on the LMFD, the proof of Proposition~\ref{prop:identifiability_breve} is now formulated as follows: 
\begin{pf}
    Condition (3) implies that $T_{wr}(q,\eta)$ is nonzero. According to Lemma~\ref{lem:LMFD}, Condition (1), (2), (5), and (6) imply that $\breve{A}(\q,\eta)$, $\breve{B}(\q,\eta)$, and $\breve{G}(\q,\eta)$ are unique up to a scalar factor $\alpha\in\mathbb{R}_+$. Condition (4) implies that the parameters cannot be scaled and therefore, $\alpha$ is fixed. As $\breve{A}(\q,\eta)$ is uniquely found, the spectral factorisation of $\Phi_{\bar{v}}(\omega,\eta)$ gives unique $\breve{F}(q,\eta)$ and $\breve{\Lambda}(\eta)$ \cite{YOULA1961}.
\end{pf}

\section{Proof of Proposition~\ref{prop:identifiability}} \label{app:prop_idf}
The conditions for network identifiability of the original mixed network model \eqref{eq:model_original} are derived using the network identifiability conditions for the polynomial model \eqref{eq:model_poly} formulated in Proposition~\ref{prop:identifiability_breve}. In addition, we make use of the uniqueness of the mapping between the two representations. 

First, the conditions of Lemma~\ref{lem:LMFD} are translated into conditions on the original mixed network model set $\mathcal{M}$.
\begin{lemma}[LMFD] \label{lem:LMFD2}
  The conditions of Lemma~\ref{lem:LMFD} are satisfied if the following conditions are satisfied:
  \begin{enumerate}
    \item $\begin{bmatrix}A(\q) & N_G(\q)\end{bmatrix}$ and $B(\q)$ left coprime in $\mathcal{M}$.
    \item A permutation matrix $P$ exists such that within $\mathcal{M}$, $\begin{bmatrix}\bar{\Upsilon}_0 & \bar{\Upsilon}_1 & \cdots & \bar{\Upsilon}_{n_{\upsilon}} & B_0 & B_1 & \cdots & B_{n_b}\end{bmatrix} P = \begin{bmatrix}D & R\end{bmatrix}$ with $D$ square, diagonal, and full rank.
    \item The zero structure of $G(q)$ is known in $\mathcal{M}$.
    \item $G_{ji}(q)=0$ if $G_{ij}(q)\neq0$, $\forall i,j$, within $\mathcal{M}$.
  \end{enumerate}
\end{lemma}
\begin{pf}
  Condition (1) implies Condition (1) of Lemma~\ref{lem:LMFD}, because $d_G(\q)$ is exactly the common factor between $\breve{A}(\q)$ and $\breve{B}(\q)$ within $\breve{\mathcal{M}}$ if $A(\q)$ and $B(\q)$ are left coprime within $\mathcal{M}$. $d_G(\q)$ is constructed such that all common factors of $D_G(\q)$ and $N_G(\q)$ are taken out, i.e. there are no common terms created in $G(q)$. Therefore, $\breve{G}(\q)$ and $\breve{B}(\q)$ are left coprime within $\breve{\mathcal{M}}$ if $N_G(\q)$ and $B(\q)$ are left coprime within $\mathcal{M}$. Condition (2)-(4) imply Condition (2)-(4) of Lemma~\ref{lem:LMFD}, because $\breve{A}(\q)$, $\breve{G}(\q)$, and $\breve{B}(\q)$ adopt the zero structure of $A(\q)$, $N_G(\q)$, and $B(\q)$, respectively, where the zero structure of $N_G(\q)$ is the same as the one of $G(q)$.
\end{pf}

Next, an intermediate result on a unique mapping is formulated that will be used to map the polynomial matrices $\breve{A}(\q)$, $\breve{B}(\q)$, and $\breve{G}(\q)$ into $A(\q)$, $B(\q)$, and $G(q)$.
\begin{lemma}[Uniqueness] \label{lem:unique}
    Let $A_T(\q)\in\mathcal{A}$ and $B_T(\q)\in\mathcal{B}$ be left coprime. Let $\breve{A}_T(\q) := d_T(\q) A_T(\q)$, $\breve{B}_T(\q) := d_T(\q) B_T(\q)$, and $\breve{C}_T(\q) := d_T(\q)C_T(q)$, with $C_T(q)$ a rational matrix with proper elements and $d_T(\q)$ a monic polynomial.
    If there is at least one linear constraint on the coefficients of $A_T(\q,\theta_{A_T})$ and $B_T(\q,\theta_{B_T})$ of the form $\Gamma\bar{\theta}=\gamma\neq0$, with $\Gamma$ full row rank and with $\bar{\theta}^{\top} := \begin{bmatrix}\theta_{A_T}^{\top}&\theta_{B_T}^{\top}\end{bmatrix}$, then $A_T(\q)$, $B_T(\q)$, and $C_T(q)$ are uniquely obtained from scaled versions of $\breve{A}_T(\q)$, $\breve{B}_T(\q)$, and $\breve{C}_T(\q)$ through
  \begin{align*}
    \tilde{A}_T(\q) &= \alpha_T d_T(\q) A_T(\q),\\
    \tilde{B}_T(\q) &= \alpha_T d_T(\q) B_T(\q),\\
    \tilde{C}_T(\q) &= \alpha_T d_T(\q) C_T(q). 
  \end{align*}
  with $\tilde{A}_T(\q) := \alpha_T \breve{A}_T(\q)$, $\tilde{B}_T(\q) := \alpha \breve{B}_T(\q)$, $\tilde{C}_T(\q) := \alpha \breve{C}_T(\q)$, and scaling factor $\alpha\in\mathbb{R}_+$. 
\end{lemma}
\begin{pf}
  As $A_T(\q)$ and $B_T(\q)$ are left coprime, the polynomial $d_T(\q)$ is the greatest monic common divisor of $\tilde{A}_T(\q)$ and $\tilde{B}_T(\q)$, which can be found uniquely \cite{KAILATH1980}. The scalar factor $\alpha_T$ is fixed with the parameter constraint.
\end{pf}

Finally, using Lemmas~\ref{lem:LMFD},~\ref{lem:LMFD2}, and~\ref{lem:unique}, the proof for the network identifiability in Proposition~\ref{prop:identifiability} is formulated. 
\begin{pf}
    Condition (3) implies that $T_{wr}(q,\theta)=T_{wr}(q,\eta)$ is nonzero. According to Lemmas~\ref{lem:LMFD} and~\ref{lem:LMFD2}, Condition (1), (2), (5), and (6) imply that $\breve{A}(\q,\eta)$, $\breve{B}(\q,\eta)$, and $\breve{G}(\q,\eta)$ are unique up to a scalar factor $\alpha\in\mathbb{R}_+$. According to Lemma~\ref{lem:unique}, Condition (4) implies that $A(\q,\theta)$, $B(\q,\theta)$, and $G(q,\theta)$ are uniquely found from $\alpha \breve{A}(\q,\eta)$, $\alpha \breve{B}(\q,\eta)$, and $\alpha \breve{G}(\q,\eta)$. It also implies that $\alpha d_G(\q,\eta)$ is uniquely determined, which leads to $\alpha$ as $d_G(\q,\eta)$ is monic. As $A(\q,\theta)$ is uniquely found, the spectral factorisation of $\Phi_{\bar{v}}(\omega,\theta)$ gives unique $F(q,\theta)$ and $\Lambda(\theta)$ \cite{YOULA1961}.
\end{pf}

\section{Proof of Proposition~\ref{prop:datainformativity}} \label{app:prop_datainf}
The premise of implication \eqref{eq:datainf} is satisfied if and only if
$\Delta_W(q,\theta):=W(q,\theta_1)-W(q,\theta_2)=0$. Applying Parseval's theorem gives
\begin{equation*}
  \frac{1}{2\pi} \int_{-\pi}^{\pi} \Delta_{W}(e^{i\omega},\theta) \Phi_z(\omega) \Delta_{W}^{\top}(e^{-i\omega},\theta) d\omega = 0.
\end{equation*}
This implies $\Delta_{W}(q,\theta)=0$ only if $\Phi_z(\omega)\succ0$ for a sufficiently high number of frequencies. In the situation $K\geq1$, $w(t)$ depends on $r(t)$ and substituting \eqref{eq:model_original} (or, equivalently, \eqref{eq:model_poly}) for $w(t)$ gives
\begin{equation*}
  z(t) = J(q) \kappa(t).
\end{equation*}
with $\kappa^{\top}(t) = \begin{bmatrix}e^{\top}(t)&r^{\top}(t)\end{bmatrix}$ and
\begin{equation*}
  J(q) = \begin{bmatrix}\big(A(\q)-G(q)\big)^{-1}&0\\0&I\end{bmatrix}\begin{bmatrix}F(q)& B(\q) \\0&I\end{bmatrix}.
\end{equation*}
As $J(q)$ has always full rank, $\Phi_z(\omega)\succ0$ if and only if $\Phi_{\kappa}(\omega)\succ0$. As $e(t)$ and $r(t)$ are assumed to be uncorrelated and $E\{e(t)\}=0$, we have that $\Phi_{re}=\Phi_{er}=0$ and
\begin{equation*}
  \Phi_{\kappa} = \begin{bmatrix}\Phi_r&\Phi_{re}\\ \Phi_{er}&\Phi_e\end{bmatrix} = \begin{bmatrix}\Phi_r&0\\0& \Lambda\end{bmatrix}.
\end{equation*}
Then $\Phi_{\kappa}(\omega)\succ0$ if and only if $\Lambda\succ0$ (which is assumed) and $\Phi_r(\omega)\succ0$. The condition $\Phi_z(\omega)\succ0$ reduces to $\Phi_r(\omega)\succ0$.

\section{Proof of Theorem~\ref{thm:consistency}} \label{app:thm_consistency}
Convergence of the criterion \eqref{eq:criterion_theta} follows directly from applying \cite[Theorem 2B.1]{LJUNG1999} and the fact that $S\succ0$ as the conditions for convergence are satisfied by the network model set. Second, by Condition (1), $\theta^0$ is a minimum of the criterion \eqref{eq:criterion_theta}, which can be seen as follows: Substituting \eqref{eq:model_original} for $w(t)$ into the prediction error $\breve{\varepsilon}(t,\theta) = w(t) - W(q,\theta) z(t)$ leads to (omitting argument $q$)
\begin{equation} \label{eq:predictionerror_result2}
  \breve{\varepsilon}(t,\theta) = W_{\breve{\varepsilon}r}(\theta)r(t) + W_{\breve{\varepsilon}e}(\theta)e(t) + (\Upsilon_0^0)^{-1}e(t),
\end{equation}
with 
\begin{align*}
  W_{\breve{\varepsilon}r}(q,\theta) &= \Upsilon_0^{-1}(\theta) F^{-1}(\theta) \left[ \Upsilon(\theta) (\Upsilon^0)^{-1} B^0 - B(\theta) \right], \\
  W_{\breve{\varepsilon}e}(q,\theta) &= \Upsilon_0^{-1}(\theta) F^{-1}(\theta) \Upsilon(\theta) (\Upsilon^0)^{-1} F^0 - (\Upsilon_0^0)^{-1},
\end{align*}
where $\Upsilon_0 = A_0-G^{\infty}$. As $r(t)$ and $e(t)$ are uncorrelated and $W_{\bar{\varepsilon}e}(q,\theta)$ is strictly proper, the power of any cross term between the three terms in the prediction error \eqref{eq:predictionerror_result2} is zero, so the power of each term can be minimised individually. As a result, $W_{\bar{\varepsilon}r}(q,\theta^0)=0$ and $W_{\bar{\varepsilon}e}(q,\theta^0)=0$ and thus the cost function reaches its minimum value when the prediction error is equal to the innovation $(\Upsilon_0^0)^{-1}e(t)$. Third, following the result \cite[Theorem 8.3]{LJUNG1999}, under Condition (2), the minimum at $\theta^0$ provides a unique predictor filter $W(q,\theta)$ and therefore also a unique pair $(T_{wr}(q,\theta),\Phi_{\bar v}(\omega,\theta))$. With Condition (3) this implies that the resulting model $M(\theta)=M(\theta^0)$ is unique. Therefore, $M(\tn)$ converges to $M(\theta^0)$ with probability 1.

\bibliographystyle{plain}
\bibliography{LibraryMixed2024}

@ARTICLE{YOULA1961,
  author = {D. Youla},
  journal = {IRE Transactions on Information Theory},
  title = {On the factorization of rational matrices},
  year = {1961},
  volume = {7},
  number = {3},
  pages = {172-189},
}

@BOOK{KAILATH1980,
  author = {T. Kailath},
  title = {Linear systems},
  publisher = {Prentice-Hall},
  year = {1980},
  address = {Englewood Cliffs, NJ},
}

@INCOLLECTION{JONES1985,
  author = {M.B. Jones},
  title = {Plant microclimate},
  editor = {J. Coombs and D.O. Hall and S.P. Ljong and J.M.O. Scurlock},
  booktitle = {Techniques in Bioproductivity and Photosynthesis},
  edition = {Second},
  publisher = {Pergamon},
  year = {1985},
  pages = {26-40},
  chapter = {3},
  series = {Pergamon International Library of Science, Technology, Engineering and Social Studies},
  isbn = {978-0-08-031999-5},
}

@BOOK{LJUNG1994,
  author = {L. Ljung and T. Glad},
  title = {Modeling of dynamic systems},
  publisher = {Prentice-Hall},
  year = {1994},
  address = {Englewood Cliffs, NJ},
  isbn = {978-0-13-597097-0},
}

@BOOK{LJUNG1999,
  author = {L. Ljung},
  title = {System identification: theory for the user},
  publisher = {Prentice-Hall},
  year = {1999},
  address = {Englewood Cliffs, NJ},
}

@ARTICLE{LUS2003,
  author = {H. Lu\c{s} and M. De Angelis and R. Betti and R.W. Longman},
  title = {Constructing second-order models of mechanical systems from identified state space realizations. part I: theoretical discussions},
  journal = {Journal of Engineering Mechanics},
  year = {2003},
  volume = {129},
  number = {5},
  pages = {477-488},
}

@INPROCEEDINGS{REN2005,
  author = {W. Ren and R. W. Beard and E. M. Atkins},
  title = {A survey of consensus problems in multi-agent coordination},
  booktitle = {Proceedings of the 2005 American Control Conference},
  year = {2005},
  pages = {1859-1864},
}

@ARTICLE{BOCCALETTI2006,
  author = {S. Boccaletti and V. Latora and Y. Moreno and M. Chavez and D.-U. Hwang},
  title = {Complex networks: structure and dynamics},
  journal = {Physics Reports},
  year = {2006},
  volume = {424},
  number = {4},
  pages = {175-308},
}

@ARTICLE{GONCALVES2008,
  author={J. Gon\c{c}alves and S. Warnick},
  title={Necessary and Sufficient Conditions for Dynamical Structure Reconstruction of LTI Networks},
  journal={IEEE Transactions on Automatic Control},
  year={2008},
  volume={53},
  number={7},
  pages={1670-1674},
}

@BOOK{MESBAHI2010,
  author = {M. Mesbahi and M. Egerstedt},
  title = {Graph theoretic methods in multiagent networks},
  publisher = {Princeton University Press},
  year = {2010},
}

@BOOK{HANNAN2012,
  author = {E.J. Hannan and M. Deistler},
  title = {The statistical theory of linear systems},
  publisher = {Society for Industrial and Applied Mathematics (SIAM)},
  year = {2012},
  address = {Philadelphia},
}

@BOOK{VERHAEGEN2012,
  author = {M. Verhaegen and V. Verdult},
  title = {Filtering and system identification: a least squares approach},
  publisher = {Cambridge University Press},
  year = {2012},
  address = {Cambridge},
  isbn = {978-1-107-40502-8},
}

@ARTICLE{VANDENHOF2013,
  author = {P.M.J. {Van den Hof} and A. Dankers and P.S.C. Heuberger and X. Bombois},
  title = {Identification of dynamic models in complex networks with prediction error methods-Basic methods for consistent module estimates},
  journal = {Automatica},
  year = {2013},
  volume = {49},
  number = {10},
  pages = {2994-3006},
}

@ARTICLE{DORFLER2013,
  author = {F. D\"{o}rfler and F. Bullo},
  title = {Kron Reduction of Graphs With Applications to Electrical Networks},
  journal = {IEEE Transactions on Circuits and Systems I: Regular Papers},
  year = {2013},
  volume = {60},
  number = {1},
  pages = {150-163},
}

@ARTICLE{DANKERS2014,
  author = {A. Dankers and P.M.J. {Van den Hof} and X. Bombois and P.S.C. Heuberger},
  title = {Errors-in-variables identification in dynamic networks by an instrumental variable approach},
  journal = {IFAC Proceedings Volumes},
  year = {2014},
  volume = {47},
  number = {3},
  pages = {2335-2340},
  note = {19th IFAC World Congress},
}

@ARTICLE{HABER2014,
  author = {A. Haber and M. Verhaegen},
  title = {Subspace identification of large-scale interconnected systems},
  journal = {IEEE Transactions on Automatic Control},
  volume = {59},
  number = {10},
  pages = {2754-2759},
  year = {2014},
}

@INPROCEEDINGS{LOPES2015,
  author = {P. {Lopes dos Santos} and J.A. Ramos and T. Azevedo-Perdico\'{u}lis and J.L. {Martins de Carvalho}},
  title = {Deriving mechanical structures in physical coordinates from data-driven state-space realizations},
  booktitle = {Proceedings of the 2015 American Control Conference (ACC)},
  year = {2015},
  pages = {1107-1112},
}

@ARTICLE{WEERTS2018_AUTO1,
  author = {H.H.M. Weerts and P.M.J. {Van den Hof} and A.G. Dankers},
  title = {Identifiability of linear dynamic networks},
  journal = {Automatica},
  year = {2018},
  volume = {89},
  pages = {247-258},
}

@ARTICLE{WEERTS2018_AUTO2,
  author = {H.H.M. Weerts and P.M.J. {Van den Hof} and A.G. Dankers},
  title = {Prediction error identification of linear dynamic networks with rank-reduced noise},
  journal = {Automatica},
  year = {2018},
  volume = {98},
  pages = {256-268},
}

@ARTICLE{GALRINHO2019,
  author = {M. Galrinho and C.R. Rojas and H. Hjalmarsson},
  title = {Parametric identification using weighted null-space fitting},
  journal = {IEEE Transactions on Automatic Control},
  year = {2019},
  volume = {64},
  number = {7},
  pages = {2798-2813},
}

@ARTICLE{HENDRICKX2019,
  author = {J.M. Hendrickx and M. Gevers and A.S. Bazanella},
  title = {Identifiability of dynamical networks with partial node measurements},
  journal = {IEEE Transactions on Automatic Control},
  year = {2019},
  volume = {64},
  number = {6},
  pages = {2240-2253},
}

@INPROCEEDINGS{KIVITS2019,
  author = {E.M.M. Kivits and P.M.J. {Van den Hof}},
  title = {A dynamic network approach to identification of physical systems},
  booktitle = {Proceedings of the 58th IEEE Conference on Decision and Control (CDC)},
  year = {2019},
  month = {December},
  pages = {4533-4538},
}

@INPROCEEDINGS{LEGAT2020,
  author = {A. Legat and J.M. Hendrickx},
  title = {Local networkiIdentifiability with partial excitation and measurement},
  booktitle = {Proceedings of the 59th IEEE Conference on Decision and Control (CDC)},
  year = {2020},
  month = {December},
  pages = {4342-4347},
}

@ARTICLE{MATERASSI2020,
  author = {D. Materassi and M.V. Salapaka},
  title = {Signal selection for estimation and identification in networks of dynamic systems: a graphical model approach},
  journal = {IEEE Transactions on Automatic Control},
  year = {2020},
  volume = {65},
  number = {10},
  pages = {4138-4153},
}

@ARTICLE{VANWAARDE2020,
  author = {H.J. {van Waarde} and P. Tesi and M.K. Camlibel},
  title = {Necessary and sufficient topological conditions for identifiability of dynamical networks},
  journal = {IEEE Transactions on Automatic Control},
  year = {2020},
  volume = {65},
  number = {11},
  pages = {4525-4537},
}

@ARTICLE{YU2020,
  author = {C. Yu and J. Chen and S. Li and M. Verhaegen},
  title = {Identification of affinely parameterized state–space models with unknown inputs},
  journal = {Automatica},
  year = {2020},
  volume = {122},
  pages = {109271},
}

@BOOK{BULLO2022,
  author = {F. Bullo},
  title = {Lectures on network systems},
  year = {2022},
  Note = {Edition 1.6},
  publisher = {Kindle Direct Publishing},
  isbn = {978-1-986425-64-3},
  url = {motion.me.ucsb.edu/book-lns},
}

@ARTICLE{CHENG2022,
  author = {X. Cheng and S. Shi and P.M.J. {Van den Hof}},
  title = {Allocation of excitation signals for generic identifiability of linear dynamic networks},
  journal = {IEEE Transactions on Automatic Control},
  year = {2022},
  volume = {67},
  number = {2},
  pages = {692-705},
}

@ARTICLE{DREEF2022,
  author = {H.J. Dreef and S. Shi and X. Cheng and M.C.F. Donkers and P.M.J. {Van den Hof}},
  title = {Excitation allocation for generic identifiability of linear dynamic networks with fixed modules},
  journal = {IEEE Control Systems Letters},
  year = {2022},
  volume = {6},
  pages = {2587-2592},
}

@INPROCEEDINGS{KIVITS2022,
  author = {E.M.M. Kivits and P.M.J. {Van den Hof}},
  title = {Local identification in diffusively coupled linear networks},
  booktitle = {Proceedings of the 61st IEEE Conference on Decision and Control (CDC)},
  year = {2022},
  pages = {874-879},
}

@ARTICLE{ZHOU2022,
  author = {T. Zhou},
  title = {Structure identifiability of an {NDS} with {LFT} parametrized subsystems}, 
  journal = {IEEE Transactions on Automatic Control}, 
  year = {2022},
  volume = {67},
  number = {12},
  pages = {6682-6697},
}

@ARTICLE{KIVITS2023_TAC,
  author = {E.M.M. Kivits and P.M.J. {Van den Hof}},
  title = {Identification of diffusively coupled linear networks through structured polynomial models},
  journal = {IEEE Transactions on Automatic Control},
  year = {2023},
  volume = {68},
  number = {6},
  pages = {3513-3528},
}

@ARTICLE{SHI2023,
  author = {S. Shi and X. Cheng and P.M.J. {Van den Hof}},
  title = {Single module identifiability in linear dynamic networks With partial excitation and measurement}, 
  journal = {IEEE Transactions on Automatic Control}, 
  year = {2023},
  volume = {68},
  number = {1},
  pages = {285-300},
}

@INPROCEEDINGS{VIZUETE2023,
  author = {R. Vizuete and J.M. Hendrickx},
  title = {Nonlinear network identifiability: the static case},
  booktitle = {Proceedings of the 62nd IEEE Conference on Decision and Control (CDC)},
  year = {2023},
  pages = {443-448},
}

@PHDTHESIS{KIVITS2024t,
  author = {E.M.M. Kivits},
  title = {Modelling and identification of physical linear networks},
  school = {Eindhoven University of Technology},
  year = {2024},
  source = {\url{research.tue.nl/files/317188667/20240222_Kivits_hf.pdf}},
}

@INPROCEEDINGS{LIANG2025,
  author = {D. Liang and E.M.M. Kivits and M. Schoukens and P.M.J. {Van den Hof}},
  title = {A frequency-domain approach for estimating continuous-time
diffusively coupled linear networks},
  booktitle = {Proceedings of the 23rd European Control Conference (ECC)},
  year = {2025},
  volume = {},
  number = {},
  pages = {2285-2290},
}

@manual{Simulink,
  author = {{The Mathworks, Inc.}},
  title = {Simulink},
  address = {Natick, MA},
  year = {2025},
  url = {www.mathworks.com/products/simulink},
}

@manual{Simscape,
  author = {{The MathWorks, Inc.}},
  title = {Simscape},
  address = {Natick, MA},
  year = {2025},
  url = {www.mathworks.com/products/simscape},
}

\end{document}